\documentclass[12pt, draftclsnofoot, onecolumn]{IEEEtran}%

\usepackage{amsmath,amsthm,amssymb}
\usepackage{mathtools}
\usepackage{amsfonts}

\usepackage{graphicx}
\usepackage{subfigure}
\usepackage{float}

\usepackage{epstopdf}

\usepackage{tabularx}

\usepackage{cite}

\usepackage{hyperref}

\usepackage{cuted}
\usepackage{multicol}

\usepackage{stfloats}

\graphicspath{{Figures/}{Codes/}}

\newtheorem{thm}{Theorem}[section]

\newtheorem{prop}[thm]{Proposition}
\theoremstyle{definition}

\theoremstyle{remark}
\newtheorem{rem}[thm]{Remark}
\numberwithin{equation}{section}

\begin{document}

\title{Towards Analytical Electromagnetic Models for Reconfigurable Intelligent Surfaces}

\author{Tiebin~Mi,~\IEEEmembership{Member,~IEEE,}
  Jianan~Zhang,~\IEEEmembership{Student Member,~IEEE,}
  Rujing~Xiong,~\IEEEmembership{Student Member,~IEEE,}
  Zhengyu~Wang,~\IEEEmembership{Student Member,~IEEE,}
  Robert~Caiming~Qiu,~\IEEEmembership{Fellow,~IEEE,}
  \thanks{T. Mi, J. Zhang, R. Xiong, Z. Wang, and R. Qiu are with the School of Electronic Information and Communications, Huazhong University of Science and Technology, Wuhan 430074, China (e-mail: \{mitiebin,zhangjn,rujing,wangzhengyu,caiming\}@hust.edu.cn).}%
}

\maketitle

\begin{abstract}
  Physically accurate and mathematically tractable models are presented to characterize scattering and reflection properties of reconfigurable intelligent surfaces (RISs). We take continuous and discrete strategies to model a single patch and patch array and their interactions with multiple incident electromagnetic (EM) waves. The proposed models consider the effect of the incident and scattered angles, polarization features, and the topology and geometry of RISs. Particularly, a simple system of linear equations can describe the multiple-input multiple-output (MIMO) behaviors of RISs under reasonable assumptions. It can serve as a fundamental model for analyzing and optimizing the performance of RIS-aided systems in the far-field regime. The proposed models are employed to identify the advantages and limitations of three typical configurations. One important finding is that complicated beam reshaping functionality can not be endowed by popular phase compensation configurations. A possible solution is the simultaneous configurations of collecting area and phase shifting. Numerical simulations verify the effectiveness of the proposed configuration schemes.
\end{abstract}


\section{Introduction}\label{S:Introduction}
\noindent Reconfigurable intelligent surfaces (RISs) have received considerable attention for their potential to enhance the communication and sensing performance of wireless networks \cite{basar2019wireless, wu2019intelligent, wu2019towards, yuan2021reconfigurable, wymeersch2020radio}. RISs consist of an array of tunable electromagnetic (EM) unit cells \cite{liu2018programmable}. Such a tuning varies widely in the literature, among which switch diode and continuous tuning varactor based controlling are popular \cite{pei2021ris, tang2020wireless, dai2020reconfigurable, nayeri2018reflectarray}. The scattering and reflection properties of EM waves impinging upon RISs can be controlled by software \cite{wu2019towards, di2020smart, di2020reconfigurable, huang2020holographic}, depending on configuration, topology, and geometry.

For communication purposes, RISs are employed to offer coverage enhancement by signal-to-noise ratio (SNR) improvement and interference suppression. Channel estimations for RIS-aided communication were discussed in \cite{liu2020matrix, wei2021channel, zheng2020intelligent, faqiri2022physfad}. The authors of \cite{wu2019intelligent, yildirim2021hybrid, zheng2021double} showed that cooperative active and passive beamforming could significantly improve the coverage and energy consumption. The impact of discrete phase shifting was studied in \cite{di2020hybrid, zhang2020reconfigurable, zhou2020spectral}. On the other hand, for environmental awareness purposes, RISs are expected to offer high-resolution sensing and localization capabilities. RISs can be equipped with embedded sensors or active unit cells \cite{hu2018beyond1, zhang2022holographic, puglielli2015design, long2021active}. A sensing scheme was studied in \cite{song2022intelligent}, which minimized the Cram\'er-Rao lower bound for high accuracy direction estimation.

A fundamental component of analyzing and optimizing RIS-aided systems is the development of simple but accurate models to characterize the scattering and reflection properties. There are two possible methodologies to model the interactions of RISs to EM waves impinging on them.

The first class, taking a continuous perspective, treats RISs as continuous surfaces subject to appropriate electromagnetic boundary conditions \cite{ozdogan2019intelligent, najafi2020physics, bjornson2020power, di2021communication}. The field supported by RISs is strictly governed by Maxwell's equations and the associated boundary conditions. This class is often used to identify the theoretical limits of performance independent of the specific implementation of RISs. However, the strategy is not particularly useful for the analysis and design because it is difficult to establish the relationship between boundary conditions and unit cell pattern, configuration, topology, and geometry \cite{ellingson2021path}.

The second class, taking a discrete perspective, calculates the response of RISs as a summation of the scattering from all unit cells \cite{wu2019intelligent, tang2020wireless, kammoun2020asymptotic}. A widely used but greatly oversimplified model is to assume that unit cells act as diffusive scatterers able to change the amplitude and phase of the incident EM waves, i.e.,
\begin{equation}\label{E:SimpleModelUnit}
  E^s = \gamma e^{j \Omega} E^i ,
\end{equation}
where $E^i$ and $E^s$ are the incident and reflected scalar electric fields, respectively, $\gamma \in (0, 1]$ is the reflection coefficient, and $\Omega \in [0, 2 \pi)$ is the phase shifting to be optimized \cite{gradoni2021end, abeywickrama2020intelligent}.

Moreover, for the transmission model of RIS-aided systems, consider communication from a single-antenna source to a single-antenna destination as an example. Suppose the RIS consists of $N$ unit cells. The deterministic channel from the transmitter to RIS is denoted by $\mathbf{h}^{i} = [ h^i_1, \cdots, h^i_N ]^\top$. The channel between RIS and receiver is denoted by $\mathbf{h}^{s} = [ h^s_1, \cdots, h^s_N ]^\top$. The configuration is represented by a diagonal matrix $\mathbf{W} = \text{diag} \bigl( \gamma_1 e^{j \Omega_1}, \cdots, \gamma_N e^{j \Omega_N} \bigr)$. The received signal at the destination is given by the summation of the scattered fields, i.e., 
\begin{equation}\label{E:SimpleModelArray}
  y = \sum_{n=1}^{N} \gamma_n e^{j \Omega_n} h^{i}_n h^{s}_n x = {\mathbf{h}^{s}}^\top \mathbf{W} \mathbf{h}^{i} x,
\end{equation}
where $x$ is the transmitted signal at the source. 

Most widely used models, e.g., (\ref{E:SimpleModelUnit}) and (\ref{E:SimpleModelArray}), are typically rudimentary to facilitate theoretical analysis and numerical simulation. A clear disadvantage of (\ref{E:SimpleModelUnit}) is that it is not derived from the physics and electromagnetic points of view. The interaction with incident EM waves (e.g., angles, frequencies, polarizations) is entirely ignored. The main weakness of (\ref{E:SimpleModelArray}) is that it does not contain any information about the topology and geometry of RISs. The model assumes that channel states are perfectly known, which becomes quite challenging without dedicated signal processing capability at the passive unit cells.

A significant limitation of current research is the lack of physically accurate and mathematically tractable models to describe the input/output behaviors of RISs. A complete description of the fields supported by RISs at any point in space requires Maxwell's equations and the associated boundary conditions. The electromagnetic nature requires us to rethink models from \emph{first principles} to ensure performance prediction and algorithmic design build on a solid foundation.

Inspired by the two classes of methodologies mentioned above, this paper uses continuous and discrete strategies to model a single patch and patch array. Specifically, by leveraging physical optics, which approximate boundary conditions that concern fields on a closed surface \cite{balanis2012advanced}, we introduce a physically accurate formula for the electric field scattered by a rectangular metallic patch with near-zero thickness. This strategy is appropriate because unit cells below 10 GHz are mainly based on simple planar shapes \cite{tapio2021survey}, e.g., rectangular patches \cite{dai2020reconfigurable, pei2021ris, dai2019wireless, wang2019design}.

On the other hand, for patch arrays, the phase discrepancies introduced by configuration and interelement path length difference are essential for beamforming functionality. The path length differences are closely related to the topology and geometry of RISs. We derive mathematically tractable models to characterize the overall input/output behaviors of RISs by applying the superposition principle on the fields scattered by all unit cells. The proposed models are both physically rigorous and easy to calculate. They can serve as fundamental models for analyzing and optimizing the performance of RIS-aided systems.

\subsection{Contributions}
\noindent The principal contribution of this work is to provide mathematically concise models to characterize the scattering and reflection properties of RISs consistent with electromagnetic theory. We focus on analytical methods to model a single patch and patch array and their interactions with multiple incident EM waves. Perhaps the most important finding is that a simple system of linear equations can describe the multiple-input multiple-output (MIMO) behaviors of RISs under reasonable assumptions. To the best of our knowledge, no studies have been performed to investigate the MIMO behaviors of RISs. It can serve as a fundamental model for analyzing and optimizing the performance of RIS-aided systems in the far-field regime.

The development of physically motivated and mathematically tractable models facilitates theoretical analysis. The second contribution is that the proposed models identify the advantages and limitations of three typical configurations. We clarify that the anomalous mirror effect is a significant limitation of continuous phase compensation configuration, especially when multiple incident EM waves exist. Another important finding is that a possible solution to complicated beam reshaping functionality is the simultaneous configurations of collecting area and phase shifting. Such configurations have the power to redistribute the incident waves to (almost) arbitrary radiation patterns.

\subsection{Outline}
\noindent The remainder of this paper is organized as follows. We start by considering how a unit cell responds to incident EM waves in Section~\ref{S:ScatteringUnitCell}. Physical optics techniques are employed to calculate the bistatic scattered field of a single patch. Section~\ref{S:ResponseArray} is devoted to the interaction of a patch array with external illuminations. Section~\ref{S:LinearRISs} focuses on analytical models of linear RISs. One interesting finding is that, under reasonable assumptions, a system of linear equations can describe the MIMO behavior of RISs. Section~\ref{S:PerformanceAnalysis} studies the performance of linear RISs with three typical configurations. We identify that the anomalous mirror effect is the inherent limitation of continuous phase compensation configuration.

\subsection{Reproducible Research}
\noindent The simulation results can be reproduced using code available at: \url{https://github.com/mitiebin/RIS_Models.git}

\subsection{Notation}
\noindent We describe the electric and magnetic fields in a spherical coordinate system. The strength of the incident electric field is denoted by $E^i ( \theta^i,\phi^i )$. The boldface letters stand for vector fields. For example, we denote by $\mathbf{E}^s ( r^s, \theta^s, \phi^s )$ the scattered electric field at $( r^s, \theta^s, \phi^s )$. We will write it simply $\mathbf{E}^s$ when no confusion can arise. The corresponding $(r, \theta, \phi)$ components are denoted by $E^{s}_{r}$, $E^{s}_{\theta}$ and $E^{s}_{\phi}$, respectively. We denote by $\lambda$ the wavelength throughout the paper. A Cartesian coordinate system is used to describe the topology and geometry of RISs. We use $\mathbf{e}_x$, $\mathbf{e}_y$, and $\mathbf{e}_z$ to denote the unit vectors in $x$, $y$, and $z$ directions.

\section{Scattering from a single planar unit cell}\label{S:ScatteringUnitCell}
\noindent In this section, we consider the fundamental principles that govern the interaction of each unit cell with incident EM fields. To characterize the scattered fields observed from multiple angles, models taking into account the incident and reflective directions simultaneously are necessary. Borrowing the concept of bistatic scattering cross section \cite{balanis2012advanced, osipov2017modern}, we begin by calculating the bistatic scattered field of a single patch. We assume throughout the paper that the incident field consists of multiple uniform plane waves (narrow banded) along with various directions. In the far-field regime, the electric and magnetic fields at any point are perpendicular to each other and the propagation direction.

\subsection{The bistatic scattered field from a rectangular patch}\label{SS:BistaticRCS}
\noindent For the convenience of modeling, this paper assumes that unit cells are approximated well by rectangular perfect electric conducting plates \cite{dai2020reconfigurable, pei2021ris, dai2019wireless, wang2019design}. We first consider a patch illuminated by one incident EM wave originating from $(\theta^i, \phi^i)$. The strength/magnitude of the incident electric field is denoted by $E^i (\theta^i, \phi^i)$. We will write it simply $E^i$ when no confusion can arise. The incident wave is linearly polarized with perpendicular polarization, as illustrated in Fig.~\ref{Fig:Scattering_OneIncident}. The scattered field can be evaluated using physical optics techniques by neglecting the edge effects (suppose the induced current density on the patch is the same as that on an infinite conducting flat plate).

\begin{figure}[!htbp]
\centering
\subfigure[]{
\label{Fig:Scattering_OneIncident}
\includegraphics[width=0.46\columnwidth]{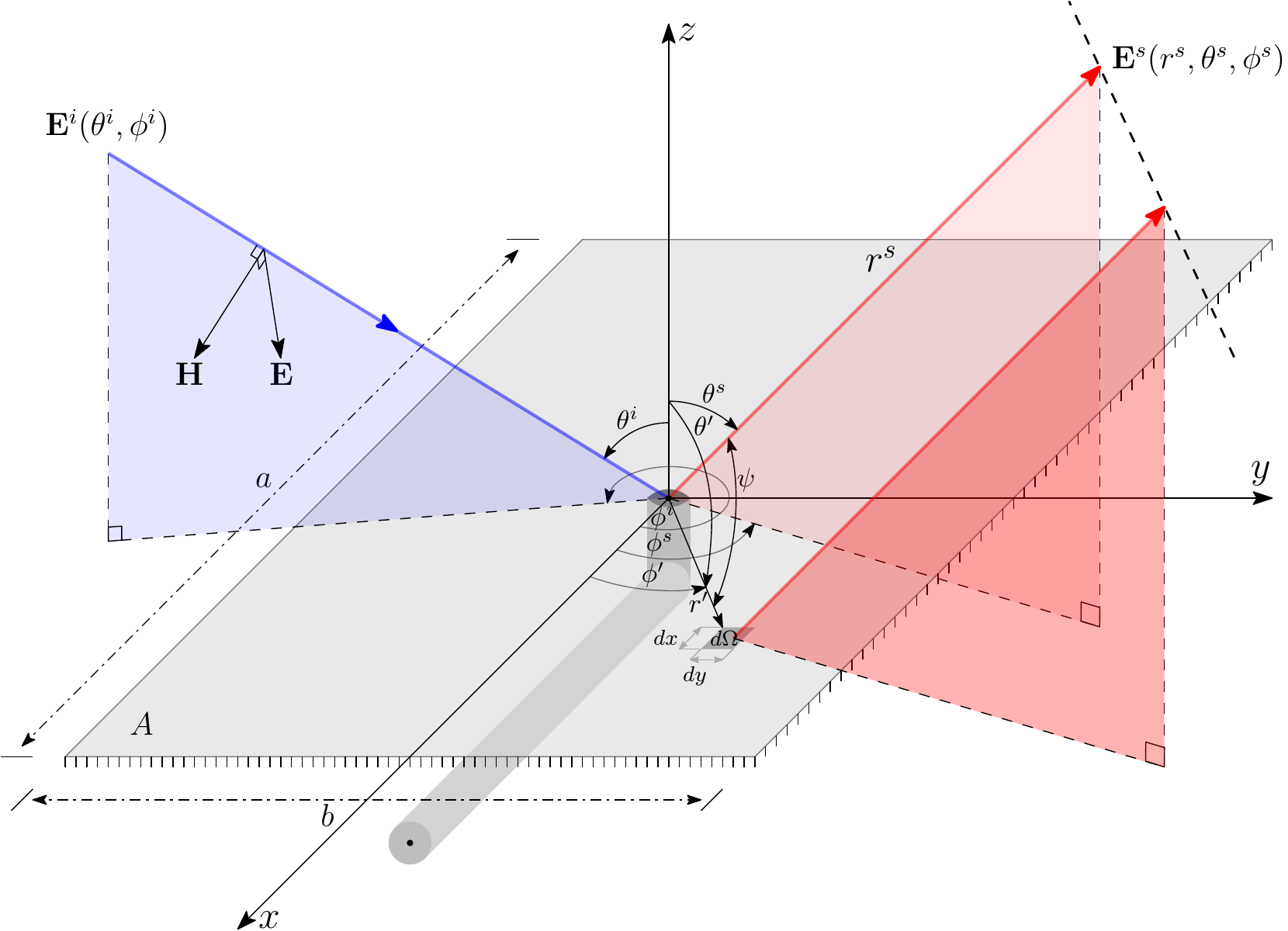}}
\hspace{0.25cm}
\subfigure[]{
\label{Fig:Scattering_MutipleIncidents}
\includegraphics[width=0.46\columnwidth]{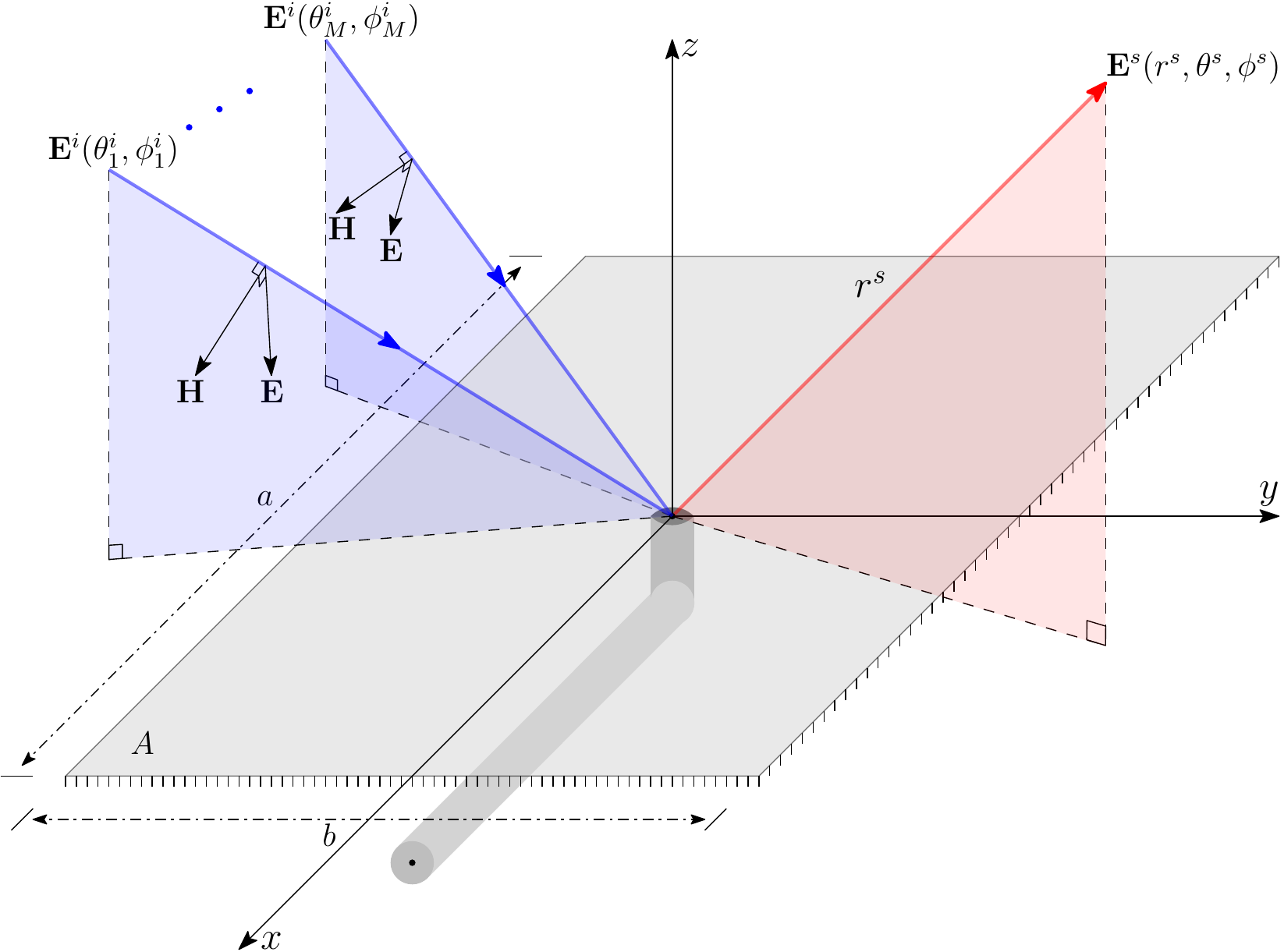}}
\caption{A rectangular metallic patch of size $a \times b$ is illuminated by EM waves, which are linearly polarized with perpendicular polarizations. The scattered electric field is observed at $(r^s, \theta^s, \phi^s)$. (a) One incident wave. (b) Muliple incident waves originating from ${ (\theta^i_1, \phi^i_1), \ldots, (\theta^i_M, \phi^i_M) }$.}
\label{Fig:Scattering}
\end{figure}

\begin{thm}\label{T:ScatterdFieldPatch}
Suppose a patch of size $a \times b$ is illuminated by one uniform plane wave $E^i (\theta^i, \phi^i)$. Let $\Gamma$ be the reflection coefficient. In the far-field regime, the $(r,\theta,\phi)$ components of the scattered electric field $\mathbf{E}^s$ at $(r^s, \theta^s, \phi^s)$ are given as
  \begin{equation}\label{E:E_s_r_theta_phi_Patch}
    \begin{dcases*}
      \begin{aligned}
      E^s_r ( r^s, \theta^s, \phi^s ) \simeq 
      & 0 , \\
      E^s_{\theta} ( r^s, \theta^s, \phi^s ) \simeq 
      & C \frac{ A }{ \lambda } \frac{ e^{-j 2 \pi r^s / \lambda } }{ r^s } E^i \cos \theta^i \cos \theta^s \bigl( \cos \phi^i \sin \phi^s - \sin \phi^i \cos\phi^s\bigr) \text{Sa} (a, b; \theta^s,\phi^s; \theta^i, \phi^i) , \\
      E^s_{\phi} ( r^s, \theta^s, \phi^s ) \simeq 
      & C \frac{ A }{ \lambda } \frac{ e^{-j 2 \pi r^s / \lambda } }{ r^s } E^i \cos \theta^i \bigl( \sin \phi^i \sin \phi^s + \cos \phi^i \cos \phi^s \bigr) \text{Sa} (a, b; \theta^s,\phi^s; \theta^i, \phi^i) ,
      \end{aligned}
    \end{dcases*}
  \end{equation}
where $A = a b$, $C = - j (1-\Gamma) / 2$, and 
\begin{equation*}
\begin{aligned}
  \text{Sa} (a, b; \theta^s,\phi^s; \theta^i, \phi^i) = 
  \frac{\sin\bigl( \frac{ \pi a }{\lambda} \bigl(\sin \theta^s \cos \phi^s + \sin \theta^i \cos \phi^i\bigr)\bigr)} { \frac{ \pi a }{\lambda} \bigl(\sin \theta^s \cos \phi^s + \sin \theta^i \cos \phi^i \bigr) } \frac{\sin \bigl( \frac{ \pi b }{\lambda} ( \sin \theta^s \sin \phi^s + \sin \theta^i \sin \phi^i ) \bigr) }{ \frac{ \pi b }{\lambda} ( \sin \theta^s \sin \phi^s + \sin \theta^i \sin \phi^i ) } .
\end{aligned}
\end{equation*} 
\end{thm}

\begin{proof}
See Appendix~\ref{S:Appendix_A}.
\end{proof}

\begin{rem}
  For the role of the term $e^{-j 2 \pi r^s / \lambda } / r^s$ in (\ref{E:E_s_r_theta_phi_Patch}), we find that $e^{-j 2 \pi r^s / \lambda }$ actually represents the propagation delay due to the distance $r^s$. It is well known that, for narrow banded signals, the propagation delays generally correspond to phase shifts. Another property of EM waves is that the electric field is expected to decrease as $1/r^s$ in free space as the distance $r^s$ increases. The term $1/r^s$ thus serves as an attenuation factor. 
\end{rem}

\begin{rem}
  As expected, the strength of the scattered electric field is proportional to plate area $A$. At the receiving stage, the patch converts the incident power density into received power by the collecting area. We emphasize that the unit cell commonly consists of sub-wavelength-scaled patterns. The actual collecting/effective area is not equal but closely related to its physical size.
\end{rem}

\begin{rem}
The term $\text{Sa} (a, b; \theta^s,\phi^s; \theta^i, \phi^i)$ essentially characterizes the intrinsic directivity of the patch. When $a \ll \lambda$ and $b \ll \lambda$, $\text{Sa} (a, b; \theta^s,\phi^s; \theta^i, \phi^i) \simeq 1$. That implies that a unit cell behaves like an isotropic point source without a strong intrinsic scattering directivity when sub-wavelength-sized (e.g., $\lambda / 5 \times \lambda / 5$ or even small). On the other hand, when $a \gg \lambda$ and $b \gg \lambda$, the patch performs specular reflection (Snell's law), although it has multiple reflection lobes. The reason is that, if we choose $\theta^s = \theta^i$ and $\phi^s = \phi^i \pm \pi$, $\text{Sa} (a, b; \theta^s,\phi^s; \theta^i, \phi^i) = 1$. 
\end{rem}

With Theorem~(\ref{T:ScatterdFieldPatch}), the strength of the scattered electric field at $(r^s, \theta^s, \phi^s)$ is 
\begin{equation}\label{E:E_TotalPatch}
  \begin{aligned}
    | \mathbf{E}^s ( r^s, \theta^s, \phi^s ) | 
    =      & \sqrt{ ( E^s_{r} ( r^s, \theta^s, \phi^s ) )^2 + ( E^s_{\theta} ( r^s, \theta^s, \phi^s ) )^2 + ( E^s_{\phi} ( r^s, \theta^s, \phi^s ) )^2 } \\
    \simeq & | C | \frac{ A }{ \lambda r^s } E^{i} ( \theta^i, \phi^i ) \cos \theta^i  \bigl| \text{Sa} (a, b; \theta^s,\phi^s; \theta^i, \phi^i) \bigr| \\
           & \sqrt{\cos^2 \theta^s \bigl( \cos \phi^i \sin\phi^s - \sin \phi^i \cos \phi^s\bigr)^2 + \bigl( \sin \phi^i \sin \phi^s + \cos \phi^i \cos \phi^s \bigr)^2 }  .
  \end{aligned}
\end{equation}
The bistatic RCS is then given as
\begin{equation}\label{E:RCS_Patch}
  \begin{aligned}
    \sigma(\theta^s, \phi^s ; \theta^i, \phi^i) 
    = & \lim_{r^s \to \infty} 4 \pi (r^s)^2 \frac{ | \mathbf{E}^s ( r^s, \theta^s, \phi^s ) |^2 }{ E^i ( \theta^i, \phi^i )^2 }                                                                                                       \\
    = & 4 \pi |C|^2 \bigl( \frac{ A }{ \lambda } \bigr)^2 \cos^2 \theta^i \bigl[ \cos^2 \theta^s \bigl( \cos \phi^i \sin\phi^s - \sin \phi^i \cos \phi^s\bigr)^2 \\
    & + \bigl( \sin \phi^i \sin \phi^s + \cos \phi^i \cos \phi^s \bigr)^2 \bigr] \text{Sa}^2 (a, b; \theta^s,\phi^s; \theta^i, \phi^i) .
  \end{aligned}
\end{equation}

Multiple comparisons are made using CST Studio Suite to evaluate the performance of the proposed models (\ref{E:E_s_r_theta_phi_Patch}) and (\ref{E:E_TotalPatch}). CST is a high-performance package for analyzing EM components and systems. The plot data from CST are exported and reproduced in Matlab for better illustration. We demonstrate in Fig.~\ref{Fig:PatchRCS_CTS} the principal plane bistatic RCS of a square patch of size $5 \lambda \times 5 \lambda$, where $f = 1/ \lambda = 3.3 \text{GHz}$, $\theta^i = 0^\circ$ and $\phi^i = 0^\circ$. The blue curves are the RCS raised by the proposed formula (\ref{E:RCS_Patch}). The dash-dotted red curves are the RCS solved by CST. The 3D RCS plots are also illustrated in Fig.~\ref{F:RCS_3D}. The comparisons show a good agreement, especially for the side lobes close to the main lobe, which validate the effectiveness of the proposed models.

\begin{figure}[!htbp]
  \centering
  \subfigure[]{
    \label{F:Fig:PatchRCS_CTS:0Deg}
    \includegraphics[width=0.4\columnwidth]{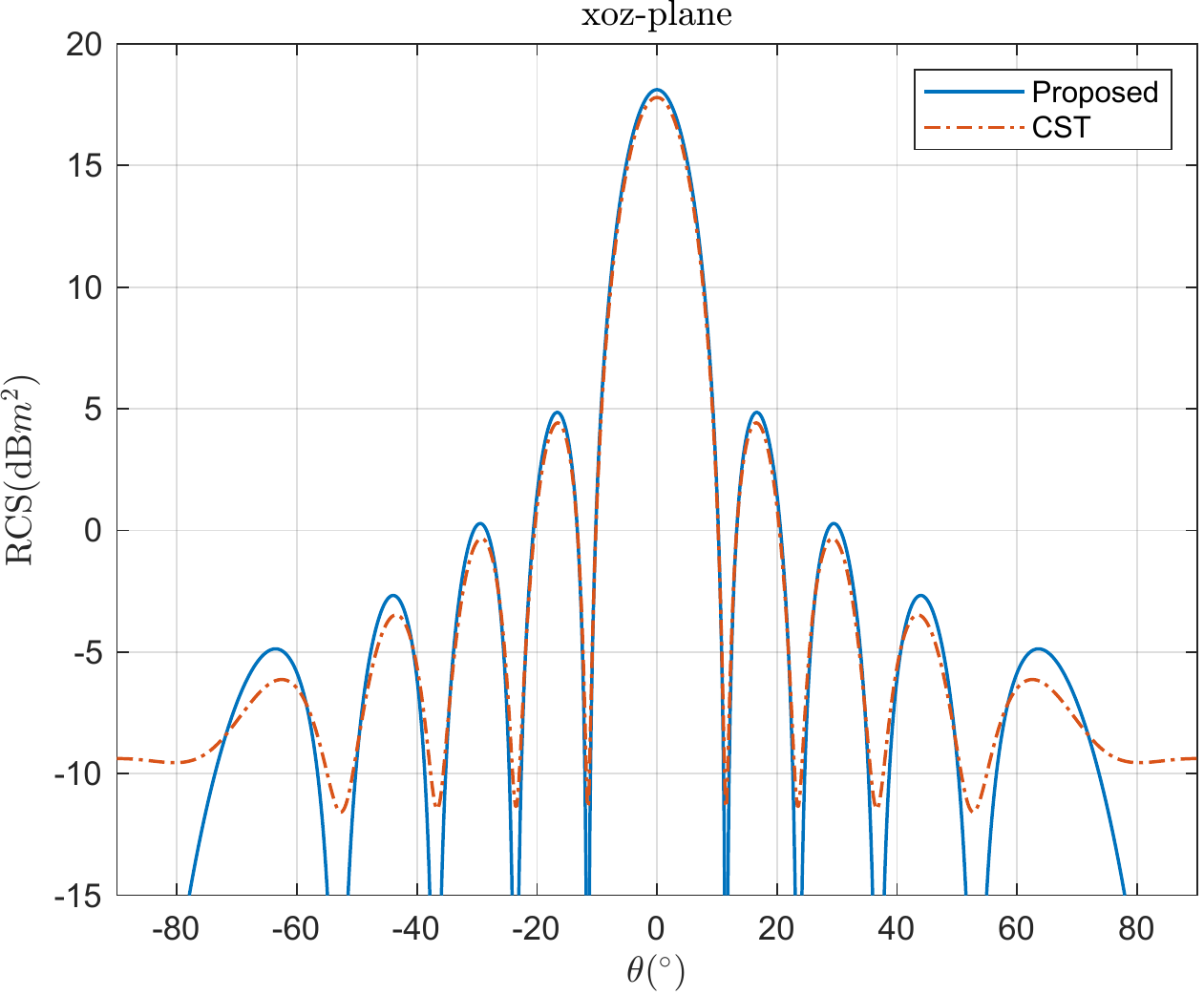}}
  \hspace{0.3cm}
  \subfigure[]{
    \label{F:Fig:PatchRCS_CTS:90Deg}
    \includegraphics[width=0.4\columnwidth]{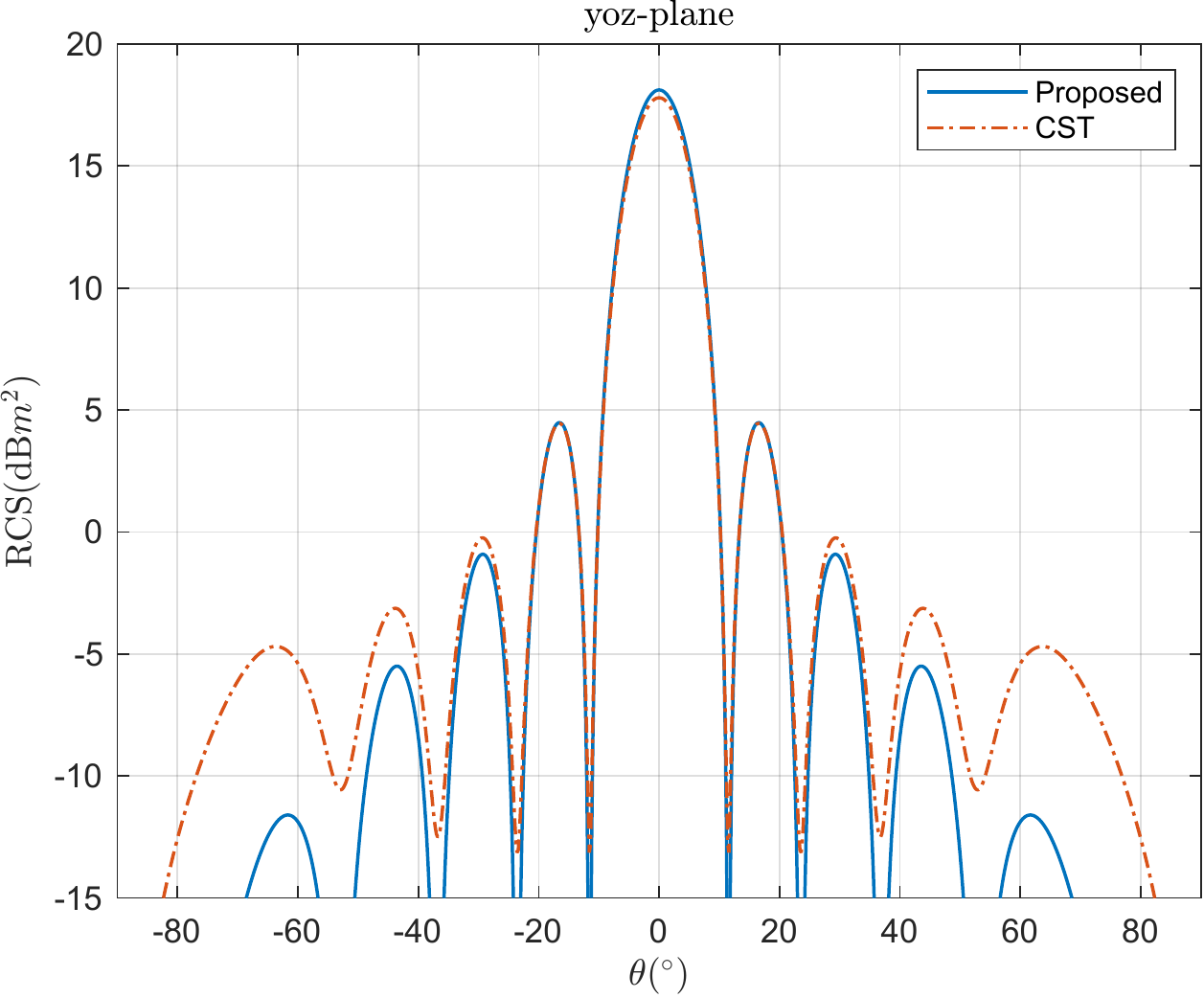}}
  \caption{Principal plane bistatic RCS of a square patch of size $5 \lambda \times 5 \lambda$. The blue curves are the RCS raised by the proposed model (\ref{E:RCS_Patch}). The dash-dotted red curves are the RCS solved by CST. The comparisons show a good agreement, especially for the side lobes close to the main lobe. (a) $xoz$-plane ($\phi^s = 0^\circ, 180^\circ$). (b) $yoz$-plane ($\phi^s = 90^\circ, 270^\circ$).}
  \label{Fig:PatchRCS_CTS}
\end{figure}

\begin{figure}[!htbp]
  \centering
  \subfigure[]{
    \label{F:F:RCS_3D:Proposed}
    \includegraphics[width=0.25\columnwidth]{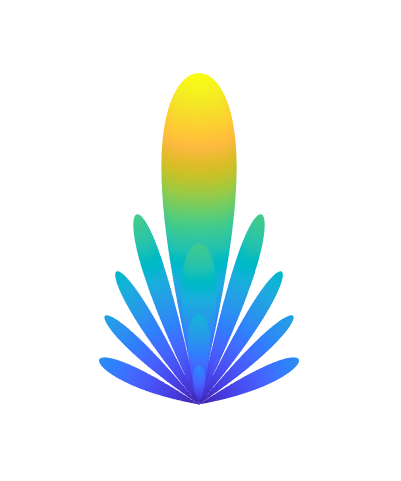}}
  \subfigure[]{
    \label{F:F:RCS_3D:CST}
    \includegraphics[width=0.25\columnwidth]{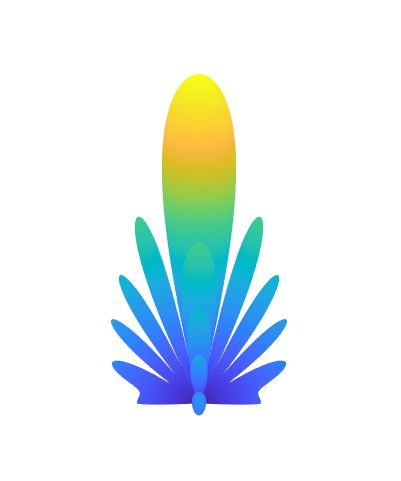}}
  \caption{The 3D bistatic RCS of a square patch. The comparisons show a good agreement, especially for the side lobes close to the main lobe. (a) The RCS raised by the proposed model (\ref{E:RCS_Patch}). (b) The RCS solved by CST.}
  \label{F:RCS_3D}
\end{figure}

\subsection{Scattering introduced by multiple incident waves}
\noindent Several EM waves may illuminate the patch simultaneously, especially when the spatially-continuous field consists of several principal incident waves, as demonstrated in Fig.~\ref{Fig:Scattering_MutipleIncidents}. The total electric current density induced on the plate can be obtained by the superposition of the individual fields related to each incident wave.

\begin{prop}
Suppose a patch of size $a \times b$ is illuminated by $M$ uniform plane waves. The strengths of the incident EM waves are $E^i (\theta^i_m, \phi^i_m), m=1, \ldots, M$. The scattered electric field observed at $(r^s, \theta^s, \phi^s)$ is given as 
\begin{equation*}
  \begin{dcases*}
    \begin{aligned}
      E^s_r ( r^s, \theta^s, \phi^s ) \simeq
       & 0 , \\
      E^s_{\theta} ( r^s, \theta^s, \phi^s ) \simeq
       & C \frac{ A }{ \lambda } \frac{ e^{-j 2 \pi r^s / \lambda } }{ r^s } \sum_{m=1}^M E^i (\theta^i_m, \phi^i_m) \cos \theta^i_m \cos \theta^s \bigl( \cos \phi^i_m \sin \phi^s - \sin \phi^i_m \cos\phi^s\bigr) \\
       & \text{Sa} (a, b; \theta^i_m, \phi^i_m; \theta^s, \phi^s) , \\
      E^s_{\phi} ( r^s, \theta^s, \phi^s ) \simeq
       & C \frac{ A }{ \lambda } \frac{ e^{-j 2 \pi r^s / \lambda } }{ r^s } \sum_{m=1}^M E^i (\theta^i_m, \phi^i_m) \cos \theta^i_m \bigl( \sin \phi^i_m \sin \phi^s + \cos \phi^i_m \cos \phi^s \bigr) \\
       & \text{Sa} (a, b; \theta^i_m, \phi^i_m; \theta^s, \phi^s) .
    \end{aligned}
  \end{dcases*}
\end{equation*}
\end{prop}

We illustrate in Fig.~\ref{F:PatchResponse} the scattered electric fields of a square patch of size $5 \lambda \times 5 \lambda$ illuminated by two incident waves originating from $(\theta^i_1 = 15^\circ, \phi^i_1 = -45^\circ)$ and $(\theta^i_2 = 45^\circ, \phi^i_2 = 135^\circ)$, respectively. Let $E^i (\theta^i_1, \phi^i_1) = 1$ and $E^i (\theta^i_2, \phi^i_2) = 0.5$. The normalized scattered electric field is plotted in Fig.~\ref{F:PatchResponse:2D}. The 3D scattered field is plotted in Fig.~\ref{F:PatchResponse:3D}. For better illustration, we also plot the incident uniform plane waves as two spikes. These results indicate that a patch of large size performs specular reflections.

\begin{figure}[!htbp]
  \centering
  \subfigure[]{
    \label{F:PatchResponse:2D}
    \includegraphics[width=0.46\columnwidth]{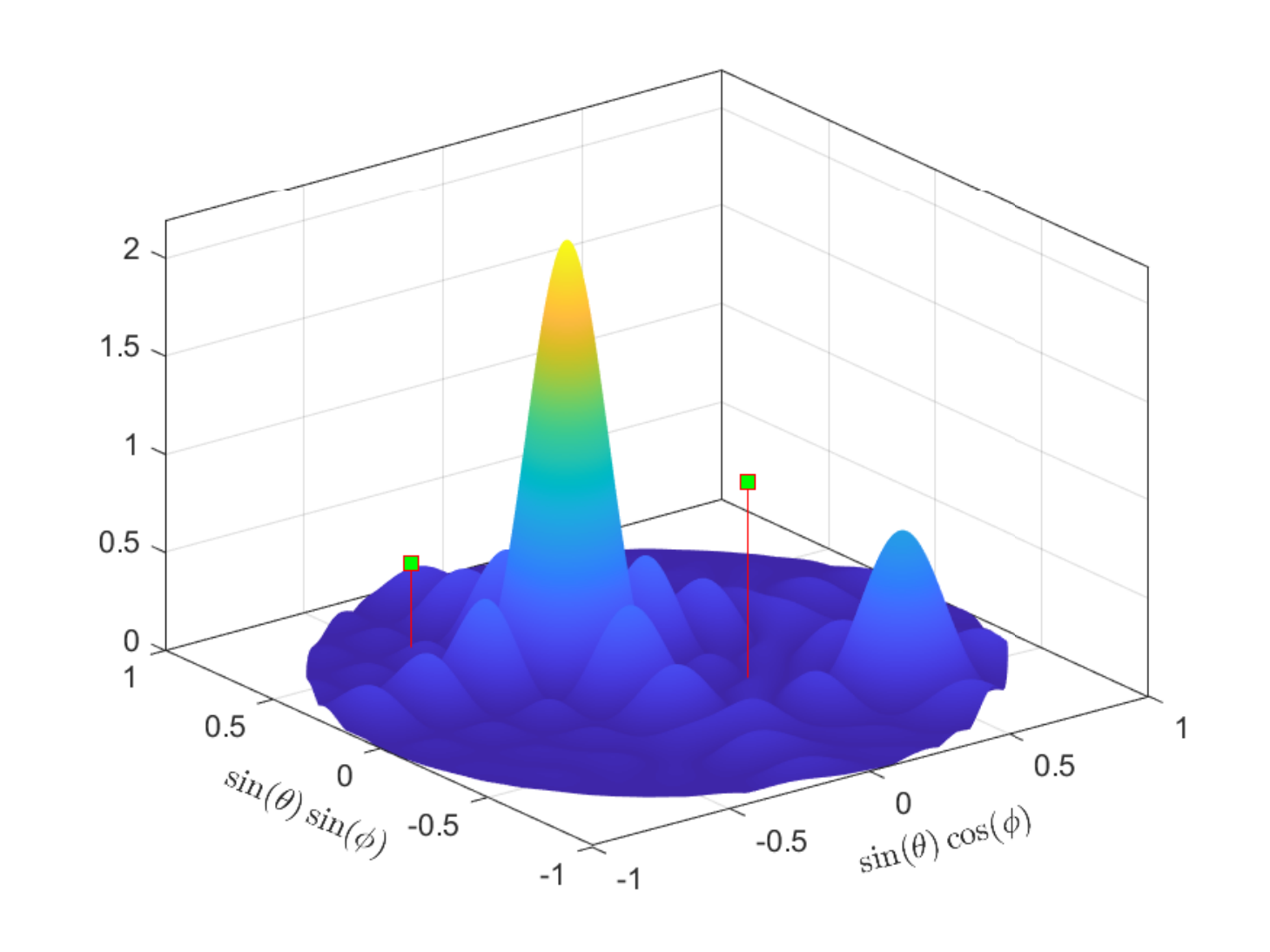}}
  \subfigure[]{
    \label{F:PatchResponse:3D}
    \includegraphics[width=0.46\columnwidth]{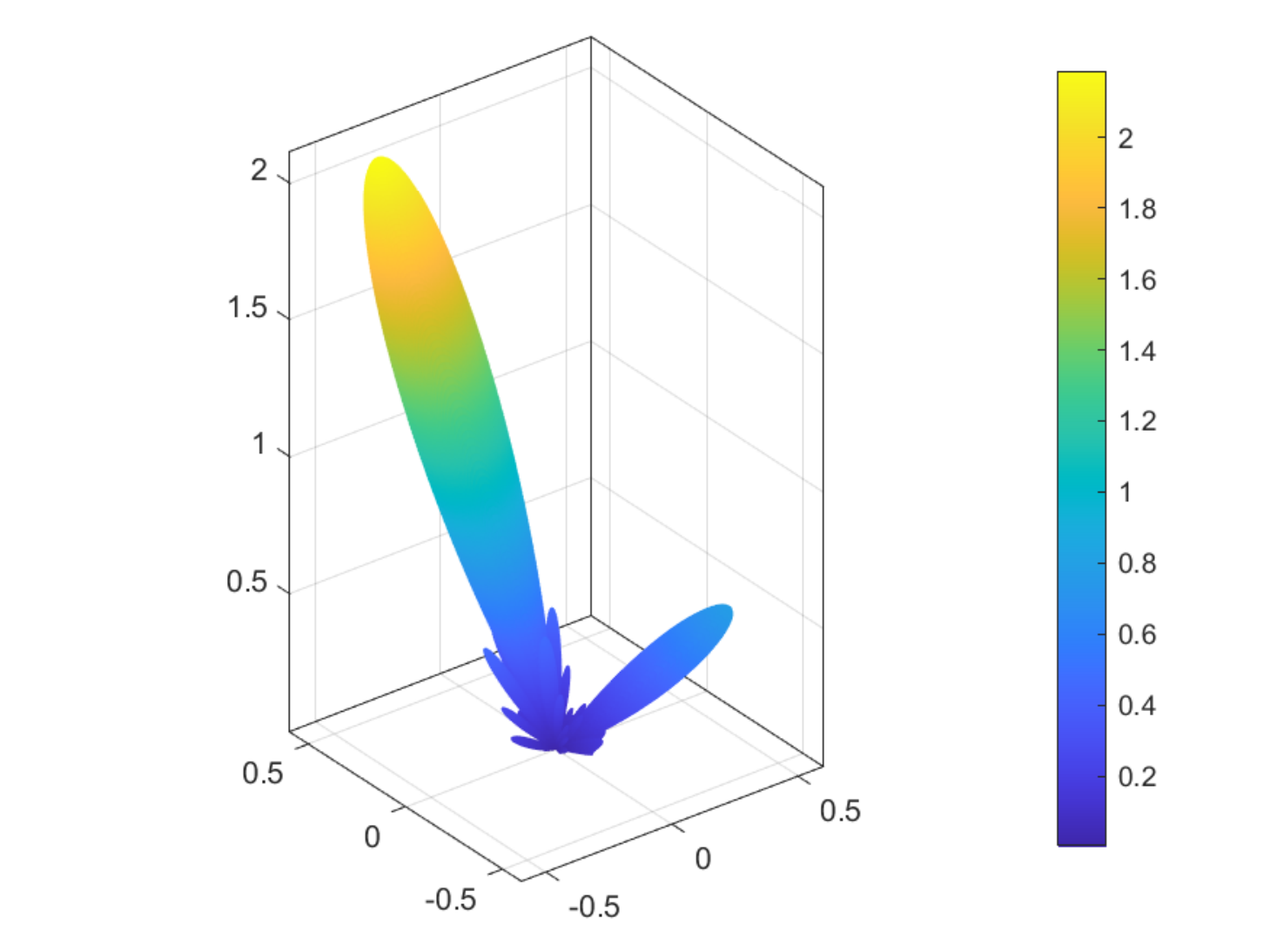}}
  \caption{The scattered electric fields of a square patch of size $5 \lambda \times 5 \lambda$ illuminated by two incident waves originating from $(\theta^i_1 = 15^\circ, \phi^i_1 = -45^\circ)$ and $(\theta^i_2 = 45^\circ, \phi^i_2 = 135^\circ)$, respectively. (a) The normalized scattered electric field. For better illustration, the two incident uniform plane waves are plotted as two spikes. (b) The 3D scattered electric field.}
  \label{F:PatchResponse}
\end{figure}

\section{Phase discrepancy due to multiple unit cells}\label{S:ResponseArray}
\noindent This section focuses on the interaction of patch arrays with external illuminations. For ease of description, we limit the discussion to RISs consisting of multiple rectangular metallic patches with near-zero thickness. We emphasize that the phase discrepancies corresponding to unit cells are essential for various beamforming functionalities. There are two equivalent ways to introduce phase discrepancies. One is the phase shifting due to the configuration. The other is the interelement path length difference, closely related to the topology and geometry of RISs.

A Cartesian coordinate system is used to describe the topology and geometry. Assume that a planar or even conformal RIS, consisting of multiple rectangular metallic patches located at $\mathbf{p}_n = [x_n, y_n, z_n ]^\top$, $n=1, \ldots, N$, is illuminated by one uniform plane wave originating from $(\theta^i, \phi^i)$. We first address the phase discrepancy related to the $n$-th unit cell, whose collecting area is $A_n$.

Suppose there is a virtual unit cell located at the origin. We denote by $\mathbf{E}_{\mathbf{o}}^s ( r^s, \theta^s, \phi^s; t )$ the scattered time-varying electric field through this virtual patch. Let $d_{\mathbf{p}_n} = d_{\mathbf{p}_n}^i + d_{\mathbf{p}_n}^s$ stand for the path length difference between the origin and $\mathbf{p}_n$, as illustrated in Fig.~\ref{Fig:Scattering_MultipleUnits_3D}. With the narrowband assumption, the scattered electric field of the unit cell at $\mathbf{p}_n$ can be written as
\begin{equation}\label{E:ScatteringDelay}
\begin{aligned}
\mathbf{E}_{\mathbf{p}_n}^s ( r^s, \theta^s, \phi^s; t ) 
= \mathbf{E}_{\mathbf{o}}^s ( r^s, \theta^s, \phi^s; t - d_{\mathbf{p}_n} / c )
\simeq \mathbf{E}_{\mathbf{o}}^s ( r^s, \theta^s, \phi^s; t ) e^{ - j 2 \pi d_{\mathbf{p}_n} / \lambda } .
\end{aligned}
\end{equation}
As discussed in Section~\ref{SS:BistaticRCS}, the complex exponential term $e^{ - j 2 \pi d_{\mathbf{p}_n} / \lambda }$ is the propagation delay created by the incident and scattered path length difference.

To calculate $d^i_{\mathbf{p}_n}$ and $d^s_{\mathbf{p}_n}$, let $\mathbf{u} (\theta,\phi) = [\sin \theta  \cos \phi, \sin \theta  \sin \phi, \cos \theta]^\top$. We can figure out 
\begin{equation*}
  \begin{dcases*}
  \begin{aligned}
    d_{\mathbf{p}_n}^{i} = & - \bigl( \mathbf{e}_x x_n + \mathbf{e}_y y_n + \mathbf{e}_z z_n \bigr) \cdot \bigl( \mathbf{e}_x \sin \theta^i \cos \phi^i + \mathbf{e}_y \sin \theta^i \sin \phi^i + \mathbf{e}_z \cos \theta^i \bigr) 
                         = - \mathbf{p}_n^\top \mathbf{u} ( \theta^{i},\phi^{i} ) , \\
    d_{\mathbf{p}_n}^{s} = & - \bigl( \mathbf{e}_x x_n + \mathbf{e}_y y_n + \mathbf{e}_z z_n \bigr) \cdot \bigl( \mathbf{e}_x \sin \theta^s \cos \phi^s + \mathbf{e}_y \sin \theta^s \sin \phi^s + \mathbf{e}_z \cos \theta^s \bigr) 
                         = - \mathbf{p}_n^\top \mathbf{u} ( \theta^{s},\phi^{s} ) .
  \end{aligned}
  \end{dcases*}
\end{equation*}
Combining (\ref{E:E_s_r_theta_phi_Patch}) and (\ref{E:ScatteringDelay}) yields the $(r, \theta, \phi)$ components of the scattered electric field of the unit cell at $\mathbf{p}_n$
\begin{equation*}
  \begin{dcases*}
    \begin{aligned}
      E^s_{\mathbf{p}_n, r} ( r^s, \theta^s, \phi^s ) \simeq & 0 , \\
      E^s_{\mathbf{p}_n, \theta} ( r^s, \theta^s, \phi^s ) \simeq 
      & C \frac{ A_n }{ \lambda } \frac{ e^{-j 2 \pi r^s / \lambda } }{ r^s } E^{i} ( \theta^i, \phi^i ) \cos \theta^i \cos \theta^s \bigl( \cos \phi^i \sin \phi^s - \sin \phi^i \cos\phi^s\bigr) \\
      & \text{Sa} (a_n, b_n; \theta^s,\phi^s; \theta^i, \phi^i) e^{ j 2 \pi \mathbf{p}_n^\top \mathbf{u} ( \theta^{i},\phi^{i} ) / \lambda } e^{ j 2 \pi \mathbf{p}_n^\top \mathbf{u} ( \theta^{s},\phi^{s} ) / \lambda } , \\
      E^s_{\mathbf{p}_n, \phi} ( r^s, \theta^s, \phi^s ) \simeq 
      & C \frac{ A_n }{ \lambda } \frac{ e^{-j 2 \pi r^s / \lambda } }{ r^s } E^{i} ( \theta^i, \phi^i ) \cos \theta^i \bigl( \sin \phi^i \sin \phi^s + \cos \phi^i \cos \phi^s \bigr) \\
      & \text{Sa} (a_n, b_n; \theta^s,\phi^s; \theta^i, \phi^i) e^{ j 2 \pi \mathbf{p}_n^\top \mathbf{u} ( \theta^{i},\phi^{i} ) / \lambda }  e^{ j 2 \pi \mathbf{p}_n^\top \mathbf{u} ( \theta^{s},\phi^{s} ) / \lambda } .
    \end{aligned}
  \end{dcases*}
\end{equation*}

\begin{figure}[!htbp]
  \centering
  \includegraphics[width=0.55\columnwidth]{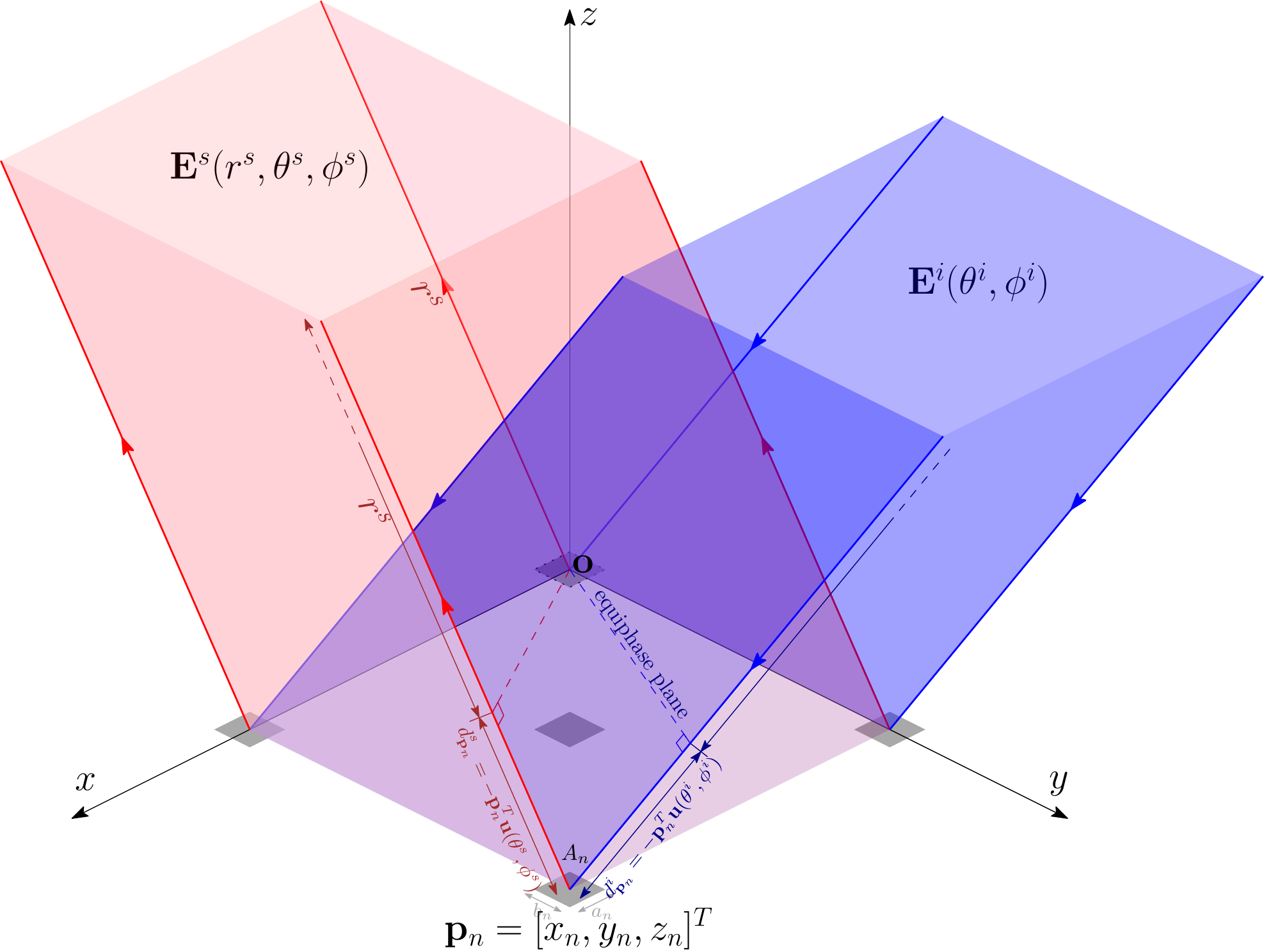}
  \caption{Phase discrepancy due to the interelement path length difference. A RIS is illuminated by one uniform plane wave (in blue) originating from $(\theta^i, \phi^i)$. The scattered electric field is observed at $(r^s, \theta^s, \phi^s)$ (in red). The incident and scattered path length differences between the origin and $\mathbf{p}_n$ are $d_{\mathbf{p}_n}^{i} = - \mathbf{p}_n^\top \mathbf{u} ( \theta^{i},\phi^{i} )$ and $d_{\mathbf{p}_n}^{s} = - \mathbf{p}_n^\top \mathbf{u} ( \theta^{s},\phi^{s} )$, respectively. }
  \label{Fig:Scattering_MultipleUnits_3D}
\end{figure}

On the other hand, the sub-wavelength-scaled pattern of each unit cell is used to change the value of the surface current density. We denote by $e^{j \Omega_n}, n=1, \ldots, N$, the phase shifting operations. The total electric field can be calculated by the superposition of the individual fields.

\begin{prop}
Suppose a RIS, consisting of multiple rectangular patches located at $\mathbf{p}_n = [x_n, y_n, z_n ]^\top$, $n=1, \ldots, N$, is illuminated by one uniform plane wave $E^i (\theta^i, \phi^i)$. The scattered electric field observed at $(r^s, \theta^s, \phi^s)$ is given as 
\begin{equation}\label{E:ScatteringArrayOneIncident}
  \begin{dcases*}
    \begin{aligned}
      E^s_{r} ( r^s, \theta^s, \phi^s ) \simeq & 0 , \\
      E^s_{\theta} ( r^s, \theta^s, \phi^s ) \simeq
                                               & C \frac{ e^{-j 2 \pi r^s / \lambda } }{ r^s }  E^{i} ( \theta^i, \phi^i ) \cos \theta^i \cos \theta^s \bigl( \cos \phi^i \sin \phi^s - \sin \phi^i \cos\phi^s\bigr)  \\
                                               & \sum_{n=1}^N \frac{A_n}{\lambda} e^{j \Omega_n} \text{Sa} (a_n, b_n; \theta^s,\phi^s; \theta^i, \phi^i) e^{ j 2 \pi \mathbf{p}_n^\top \mathbf{u} ( \theta^{i},\phi^{i} ) / \lambda }  e^{ j 2 \pi \mathbf{p}_n^\top \mathbf{u} ( \theta^{s},\phi^{s} ) / \lambda } , \\
      E^s_{\phi} ( r^s, \theta^s, \phi^s ) \simeq
                                               & C \frac{ e^{-j 2 \pi r^s / \lambda } }{ r^s }  E^{i} ( \theta^i, \phi^i ) \cos \theta^i \bigl( \sin \phi^i \sin \phi^s + \cos \phi^i \cos \phi^s \bigr) \\
                                               & \sum_{n=1}^N \frac{A_n}{\lambda} e^{j \Omega_n} \text{Sa} (a_n, b_n; \theta^s,\phi^s; \theta^i, \phi^i) e^{ j 2 \pi \mathbf{p}_n^\top \mathbf{u} ( \theta^{i},\phi^{i} ) / \lambda }  e^{ j 2 \pi \mathbf{p}_n^\top \mathbf{u} ( \theta^{s},\phi^{s} ) / \lambda }.
    \end{aligned}
  \end{dcases*}
\end{equation}
\end{prop}

The strength of the scattered electric field at $(r^s, \theta^s, \phi^s)$ is 
\begin{equation}\label{E:Strength_RIS}
  \begin{aligned}
           | \mathbf{E}^s ( r^s, \theta^s, \phi^s ) |
    =      & \sqrt{ ( E^s_{r} ( r^s, \theta^s, \phi^s ) )^2 + ( E^s_{\theta} ( r^s, \theta^s, \phi^s ) )^2 + ( E^s_{\phi} ( r^s, \theta^s, \phi^s ) )^2 } \\
    \simeq & | C | \frac{ E^{i} ( \theta^i, \phi^i ) }{ r^s } \cos \theta^i \sqrt{\cos^2 \theta^s \bigl( \cos \phi^i \sin\phi^s - \sin \phi^i \cos \phi^s\bigr)^2 + \bigl( \sin \phi^i \sin \phi^s + \cos \phi^i \cos \phi^s \bigr)^2 } \\
           & \Bigl| \sum_{n=1}^N \frac{A_n}{\lambda} e^{j \Omega_n} \text{Sa} (a_n, b_n; \theta^s,\phi^s; \theta^i, \phi^i) e^{ j 2 \pi \mathbf{p}_n^\top \mathbf{u} ( \theta^{i},\phi^{i} ) / \lambda }  e^{ j 2 \pi \mathbf{p}_n^\top \mathbf{u} ( \theta^{s},\phi^{s} ) / \lambda } \Bigr|.
  \end{aligned}
\end{equation}
The bistatic RCS is given as 
\begin{equation}\label{E:RCS_GeneralRIS}
  \begin{aligned}
    \sigma (\theta^s, \phi^s; \theta^i, \phi^i) 
    = & \lim_{r^s \to \infty} 4 \pi (r^s)^2 \frac{ | \mathbf{E}^s ( r^s, \theta^s, \phi^s ) |^2 }{ E^i ( \theta^i, \phi^i )^2 }                                                                                                                                                      \\
    \simeq & 4 \pi | C |^2 \cos^2 \theta^i \bigl[ \cos^2 \theta^s \bigl( \cos \phi^i \sin\phi^s - \sin \phi^i \cos \phi^s\bigr)^2 + \bigl( \sin \phi^i \sin \phi^s + \cos \phi^i \cos \phi^s \bigr)^2 \bigr]                                                                                    \\
      & \Bigl| \sum_{n=1}^N \frac{ A_n }{ \lambda } e^{j \Omega_n} \text{Sa} (a_n, b_n; \theta^s,\phi^s; \theta^i, \phi^i) e^{ j 2 \pi \mathbf{p}_n^\top \mathbf{u} ( \theta^{i},\phi^{i} ) / \lambda }  e^{ j 2 \pi \mathbf{p}_n^\top \mathbf{u} ( \theta^{s},\phi^{s} ) / \lambda } \Bigr|^2 .
  \end{aligned}
\end{equation}

\begin{rem}
  The proposed equation (\ref{E:Strength_RIS}) exhibits the inter-unit cooperative beamforming gain, the most crucial feature of RIS. The beamforming gain can compensate for the power loss over the distance. A general rule of thumb from (\ref{E:Strength_RIS}) is that if the scattered strength at the distance $r^s$ is of the same order as the incident strength, we must have $\lambda r^s / \sum A_n = \mathcal{O} (1)$. 
\end{rem}

When there exist multiple incident waves illuminate the RIS, combining the individual scattered fields yields the overall scattered electric field.

\begin{prop}\label{T:ScatterdFieldRISMultipleIncident}
Suppose a RIS, consisting of multiple rectangular patches located at $\mathbf{p}_n = [x_n, y_n, z_n ]^\top$, $n=1, \ldots, N$, is illuminated by multiple uniform plane waves. The strengths of the incident EM waves are $E^i (\theta^i_m, \phi^i_m), m=1, \ldots, M$. The scattered electric field observed at $(r^s, \theta^s, \phi^s)$ is given as 
\begin{equation}\label{E:ScatteringArrayMultipleIncident}
  \begin{dcases*}
    \begin{aligned}
      E^s_{r} ( r^s, \theta^s, \phi^s ) \simeq 
      & 0 , \\
      E^s_{\theta} ( r^s, \theta^s, \phi^s ) \simeq 
      & C \frac{ e^{-j 2 \pi r^s / \lambda } }{ r^s } \sum_{m=1}^{M} E^{i} ( \theta^i_m, \phi^i_m ) \cos \theta^i_m \cos \theta^s \bigl( \cos \phi^i_m \sin \phi^s - \sin \phi^i_m \cos \phi^s \bigr) \\
      & \sum_{n=1}^N \frac{ A_n }{ \lambda } e^{j \Omega_n} \text{Sa} (a_n, b_n; \theta^i_m, \phi^i_m; \theta^s, \phi^s) e^{ j 2 \pi \mathbf{p}_n^\top \mathbf{u} ( \theta^{i}_m, \phi^{i}_m ) / \lambda }  e^{ j 2 \pi \mathbf{p}_n^\top \mathbf{u} ( \theta^{s},\phi^{s} ) / \lambda } , \\
      E^s_{\phi} ( r^s, \theta^s, \phi^s ) \simeq 
      & C \frac{ e^{-j 2 \pi r^s / \lambda } }{ r^s } \sum_{m=1}^{M} E^{i} ( \theta^i_m, \phi^i_m ) \cos \theta^i_m \bigl( \sin \phi^i_m \sin \phi^s + \cos \phi^i_m \cos \phi^s \bigr)  \\
      & \sum_{n=1}^N \frac{A_n}{\lambda} e^{j \Omega_n} \text{Sa} (a_n, b_n; \theta^i_m, \phi^i_m; \theta^s, \phi^s) e^{ j 2 \pi \mathbf{p}_n^\top \mathbf{u} ( \theta^{i}_m, \phi^{i}_m ) / \lambda } e^{ j 2 \pi \mathbf{p}_n^\top \mathbf{u} ( \theta^{s},\phi^{s} ) / \lambda }.
    \end{aligned}
  \end{dcases*}
\end{equation}
\end{prop}

\begin{rem}
  Proposition~(\ref{T:ScatterdFieldRISMultipleIncident}) presents a pretty generic model for RISs consistent with electromagnetic theory in the far-field regime. The only assumption is that unit cells can be approximated well by rectangular metallic patches. Unlike other popular models, the proposed model takes into account the effect of the incident and scattered angles, polarization features, and the topology and geometry of RISs.
\end{rem}

\section{Uniform linear RISs}\label{S:LinearRISs}
\noindent We note that the physics-based models (\ref{E:E_s_r_theta_phi_Patch}) and (\ref{E:ScatteringArrayMultipleIncident}) are too vague to offer much insight. To shed new light on the input/output behaviors, this section studies uniform linear RISs consisting of patches at $\mathbf{p}_n = [0, (n-1)d, 0 ]^\top$, $n=1, \ldots, N$. Here $d$ is the unit cell spacing.

For a quick demonstration, suppose a linear RIS is illuminated by uniform plane waves lying in the $yoz$ plane, as illustrated in Fig.~\ref{Fig:LinearRIS_1d}. The electric field can be simplified considerably if the scattered waves are observed in the same plane. To reduce the ambiguity of the evaluation angles, $\theta^i, \theta^s \in [-\pi/2, \pi/2]$ are both measured against the $z$-axis in what follows. 

Let $\text{Sa}(b; \theta^s; \theta^i)$ stand for the sampling function $\frac{\sin \bigl( \frac{ \pi b }{\lambda} ( \sin \theta^s + \sin \theta^i ) \bigr) }{ \frac{ \pi b }{\lambda} ( \sin \theta^s + \sin \theta^i ) }$. We obtain from (\ref{E:E_s_r_theta_phi_Patch}) immediately the scattered field $E^s_{\mathbf{o}} ( r^s, \theta^s )$ of the patch at the origin
\begin{equation}\label{E:ScatteringUnitCell_yzPlane}
  \begin{dcases*}
    E^s_r ( r^s, \theta^s ) \simeq 0 ,        \\
    E^s_{\theta} ( r^s, \theta^s ) \simeq 0 , \\
    E^s_{\phi} ( r^s, \theta^s ) \simeq C \frac{ A }{ \lambda } \frac{ e^{-j 2 \pi r^s / \lambda } }{ r^s } E^i ( \theta^i ) \cos \theta^i \text{Sa}(b; \theta^s; \theta^i) .
  \end{dcases*}
\end{equation}
The scattered vector electric field is reduced to a scalar field. The strength is written as
\begin{align*}
  \bigl| E^s_{\mathbf{o}} (r^s, \theta^s) \bigr| 
  = \sqrt{ ( E^s_{r} )^2 + ( E^s_{\theta} )^2 + ( E^s_{\phi} )^2 } 
  \simeq \frac{ | C | }{ r^s } \frac{ A }{ \lambda } E^i ( \theta^i )  \cos \theta^i \bigl| \text{Sa} (b; \theta^s; \theta^i) \bigr|.
\end{align*}

We now turn to phase difference created by the incident and scattered path length difference. We find from Fig.~\ref{Fig:LinearRIS_1d} that, the path length difference between $\mathbf{p}_n$ and the origin is 
\[
  d_{\mathbf{p}_n} = d_{\mathbf{p}_n}^{s} + d_{\mathbf{p}_n}^{i}
\]
and
\begin{equation*}
  \begin{dcases*}
    d_{\mathbf{p}_n}^{s} = - (n-1)d \sin \theta^s , \\
    d_{\mathbf{p}_n}^{i} = - (n-1)d \sin \theta^i .
  \end{dcases*}
\end{equation*}
With the narrowband assumption, the scattered time-varying electric field is then given by
\begin{equation*}
  E^{s}_{\mathbf{p}_n} (r^s, \theta^{s})
  \simeq C \frac{ A_n }{ \lambda } \frac{ e^{-j 2 \pi r^s / \lambda } }{ r^s } E^{i} (\theta^i) \cos \theta^i \text{Sa}(b_n; \theta^s; \theta^i) e^{ j 2 \pi (n-1) d \sin \theta^{i} / \lambda }  e^{ j 2 \pi (n-1) d \sin \theta^{s} / \lambda } .
\end{equation*}
The combination of the scattered fields of all patches yields the overall field observed at $(r^s, \theta^{s})$.

\begin{prop}\label{T:LinearRIS_SISO}
Suppose a linear RIS with spacing $d$ is illuminated by one uniform plane wave whose strength is $E^{i} (\theta^i)$. The scattered electric field observed at $(r^s, \theta^s)$ is given as 
\begin{equation}\label{E:ResultingField}
  E^{s} (r^s, \theta^{s})
  \simeq C \frac{ e^{-j 2 \pi r^s / \lambda } }{ r^s } E^{i} (\theta^i) \cos \theta^i \sum_{n=1}^{N} \frac{ A_n }{ \lambda } e^{j \Omega_n} \text{Sa}(b_n; \theta^s; \theta^i) e^{ j 2 \pi (n-1) d \sin \theta^{i} / \lambda }  e^{ j 2 \pi (n-1) d \sin \theta^{s} / \lambda } .
\end{equation}
\end{prop}

\begin{rem}
Proposition~(\ref{T:LinearRIS_SISO}) provides a natural and intrinsic characterization of the single-input single-output (SISO) behaviors of linear RISs. It determines, if there is only one incident wave, what is observed at $(r^s, \theta^s)$. How much will the incident wave be attenuated (or amplified) at $(r^s, \theta^s)$?
\end{rem}

To gain further understanding of the SISO behavior of RISs, we define the steering function\footnotemark{} of linear RISs as
\begin{equation}\label{E:SteeringFunctionLinearRIS}
  \begin{aligned}
    T(\theta^s; \theta^i) 
    = & \lim_{r^s \to \infty} \frac{r^s E^{s} (r^s, \theta^{s})}{ e^{-j 2 \pi r^s / \lambda} E^{i} (\theta^i) \cos \theta^i } \\
    \simeq & C \sum_{n=1}^{N} \frac{ A_n }{ \lambda } e^{j \Omega_n} \text{Sa}(b_n; \theta^s; \theta^i) e^{ j 2 \pi (n-1) d \sin \theta^{i} / \lambda }  e^{ j 2 \pi (n-1) d \sin \theta^{s} / \lambda } .
  \end{aligned}
\end{equation}
The bistatic RCS is then given by
\begin{equation}\label{E:RCS_LinearRIS}
  \sigma (\theta^s; \theta^i) = 4 \pi \cos^2 \theta^i \bigl| T(\theta^s; \theta^i) \bigr|^2.
\end{equation}

\footnotetext{The steering function is used to describe some inherent properties of RISs. Since the term $\cos \theta^i$ in (\ref{E:ResultingField}) is related to the particular polarization features of the EM waves impinging on RISs, it is not included in the steering function.}

When there exist multiple incident waves, the superposition of individual fields at $(r^s, \theta^{s})$ leads to the overall scattered field. We thus have a proposition to characterize the multiple-input single-output (MISO) behaviors of RISs.

\begin{prop}
Suppose a linear RIS is illuminated by multiple uniform plane waves. The strengths of the incident EM waves are $E^i (\theta^i_m), m=1, \ldots, M$. The scattered electric field observed at $(r^s, \theta^s)$ is given as 
\begin{equation}\label{E:ScatteringLinearRIS}
    E^{s} (r^s, \theta^{s}) 
    \simeq C \frac{ e^{-j 2 \pi r^s / \lambda }}{ r^s }
    \sum_{m=1}^{M} E^{i}(\theta_m^i) \cos \theta_m^i \sum_{n=1}^{N} \frac{ A_n }{ \lambda } e^{j \Omega_n} \text{Sa}(b_n; \theta^s; \theta^i_m) e^{ j 2 \pi (n-1) d \sin \theta_{m}^{i} / \lambda } e^{ j 2 \pi (n-1) d \sin \theta^{s} / \lambda } .
\end{equation}
\end{prop}

\begin{figure}[!htbp]
  \centering
  \includegraphics[width=0.75\columnwidth]{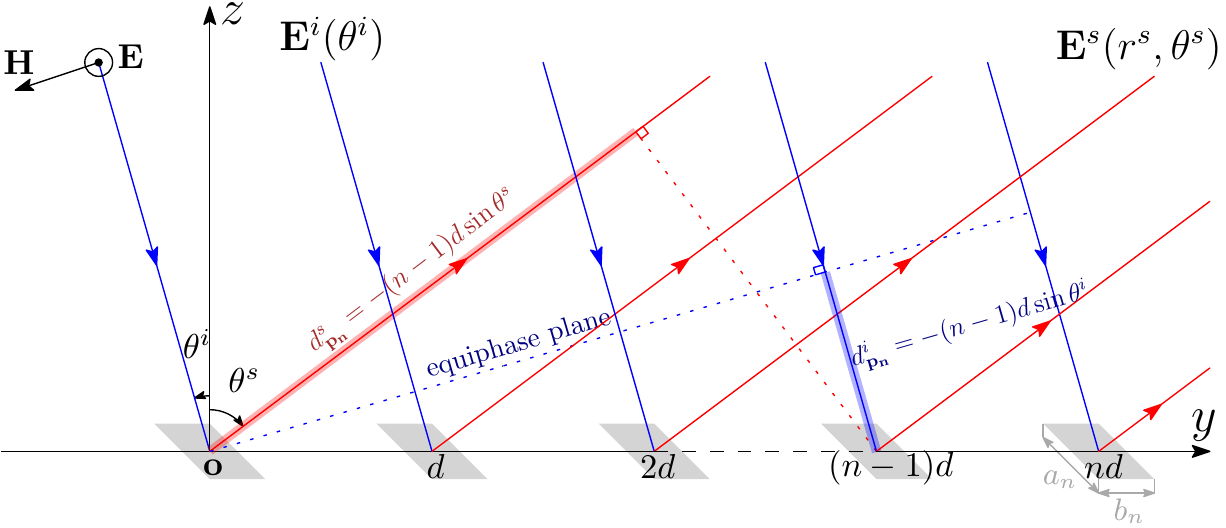}
  \caption{Phase discrepancy due to the interelement path length difference for linear RISs. The RIS is illuminated by uniform plane waves lying in the $yoz$ plane. The incident EM wave (in blue) is originating from $\theta^i$. The scattered electric field is observed at $(r^s, \theta^s)$ (in red). The incident and scattered path length differences between the origin and the unit cell at $(n-1)d$ are $d_{\mathbf{p}_n}^{i} = - (n-1)d \sin \theta^i $ and $d_{\mathbf{p}_n}^{s} = - (n-1)d \sin \theta^s$, respectively.}
  \label{Fig:LinearRIS_1d}
\end{figure}

\subsection{A canonical representation for the MIMO behaviors of RISs}
\noindent Several equations, e.g., (\ref{E:RCS_GeneralRIS}), (\ref{E:SteeringFunctionLinearRIS}), (\ref{E:RCS_LinearRIS}), and (\ref{E:ScatteringLinearRIS}), are proposed to characterize RISs' output observed at a single point. We now address how to describe the MIMO behaviors of RISs.

It is worth noting that the non-reducible term $\text{Sa} (\cdot)$, a product of two sampling functions, leads to fundamental challenges (non-linearity) to simplifying the scattering model (\ref{E:ScatteringLinearRIS}) further. Nevertheless, when $b_n \ll \lambda$, the sampling function is constant approximately, i.e., $ \text{Sa}(b_n; \theta^s; \theta^i) \simeq 1 $. The MISO representation (\ref{E:ScatteringLinearRIS}) could be expressed in matrix form
\begin{equation}\label{E:E_s_r_theta_phi_Array}
  \begin{aligned}
    E^{s} (r^s, \theta^{s})
     \simeq & \frac{ C }{ \lambda } \frac{ e^{-j 2 \pi r^s / \lambda }}{ r^s } \sum_{n=1}^{N} A_n e^{j \Omega_n} e^{ j 2 \pi (n-1) d \sin \theta^{s} / \lambda } \sum_{m=1}^{M} e^{ j 2 \pi (n-1) d \sin \theta_{m}^{i} / \lambda } \cos \theta_{m}^iE^{i} (\theta_{m}^i) \\
     = & \frac{ C }{ \lambda } \frac{ e^{-j 2 \pi r^s / \lambda }}{ r^s }
    \begin{bmatrix}
      1 & \cdots & e^{ j 2 \pi (N-1) d \sin \theta^s / \lambda }
    \end{bmatrix}
    \begin{bmatrix}
      A_1 e^{j \Omega_1} &        & 0    \\
                         & \ddots &                    \\
      0                  &        & A_N e^{j \Omega_N}
    \end{bmatrix} \\
     &
    \begin{bmatrix}
      1                                                   & \cdots & 1       \\
      \vdots                                              & \ddots & \vdots  \\
      e^{ j 2 \pi (N-1) d \sin \theta_{1}^{i} / \lambda } & \cdots & e^{ j 2 \pi (N-1) d \sin \theta_{M}^{i} / \lambda }
    \end{bmatrix}
    \begin{bmatrix}
      \cos \theta_{1}^{i} E^i (\theta_{1}^{i}) \\
      \vdots                                   \\
      \cos \theta_{M}^{i} E^i (\theta_{M}^{i})
    \end{bmatrix} .
  \end{aligned}
\end{equation}

We are now in a position to show the linear representation to characterize the MIMO behaviors of RISs. Given observation points $\{ (r_1^s, \theta_{1}^{s}), \cdots, (r_T^s, \theta_{T}^{s}) \}$, as illustrated in Fig.~\ref{Fig:ScatteringLinearRIS_MIMO}, we stack all of the observed electric fields $\{ E^{s} (r_1^s, \theta_{1}^{s}), \ldots, E^{s} (r_T^s, \theta_{T}^{s}) \}$ to form a column vector denoted by $\mathbf{E}^s (\mathbf{r}^s, \mathbf{\Theta}^s)$. With a slight abuse of notations, the incident fields $\{E^i (\theta_{1}^{i}), \ldots, E^i (\theta_{M}^{i}) \}$ are stacked into a column vector $\mathbf{E}^i (\mathbf{\Theta}^i)$. We then have a canonical representation to describe the MIMO behaviors
\begin{equation}\label{E:E_s_r_theta_phi_ArrayDense}
  \begin{aligned}
    \begin{bmatrix}
      E^{s} (r_1^s, \theta_{1}^{s}) \\
      \vdots                           \\
      E^{s} (r_T^s, \theta_{T}^{s})
    \end{bmatrix}
    = & \frac{ C }{ \lambda }
    \begin{bmatrix}
      \frac{ e^{-j 2 \pi r_1^s / \lambda }}{ r_1^s } &        & 0 \\
                                                     & \ddots & \\
       0                                             &        & \frac{ e^{-j 2 \pi r_T^s / \lambda }}{ r_T^s }
    \end{bmatrix}
    \begin{bmatrix}
      1      & \cdots & e^{ j 2 \pi (N-1) d \sin \theta_1^s / \lambda } \\
      \vdots & \ddots & \vdots                                          \\
      1      & \cdots & e^{ j 2 \pi (N-1) d \sin \theta_T^s / \lambda } \\
    \end{bmatrix} 
    \begin{bmatrix}
          A_1 e^{j \Omega_1} &        &  0 \\
                             & \ddots &    \\
          0                  &        & A_N e^{j \Omega_N}
    \end{bmatrix} \\
    & \begin{bmatrix}
      1                                                   & \cdots & 1                                                   \\
      \vdots                                              & \ddots & \vdots                                              \\
      e^{ j 2 \pi (N-1) d \sin \theta_{1}^{i} / \lambda } & \cdots & e^{ j 2 \pi (N-1) d \sin \theta_{M}^{i} / \lambda }
    \end{bmatrix} 
    \begin{bmatrix}
      \cos \theta_{1}^{i} &        &  0 \\
                          & \ddots &    \\
      0                   &        & \cos \theta_{M}^{i}
    \end{bmatrix}
    \begin{bmatrix}
      E^i (\theta_{1}^{i}) \\
      \vdots    \\
      E^i (\theta_{M}^{i})
    \end{bmatrix} .
  \end{aligned}
\end{equation}

\begin{figure}[!htbp]
  \centering
  \includegraphics[width=0.7\columnwidth]{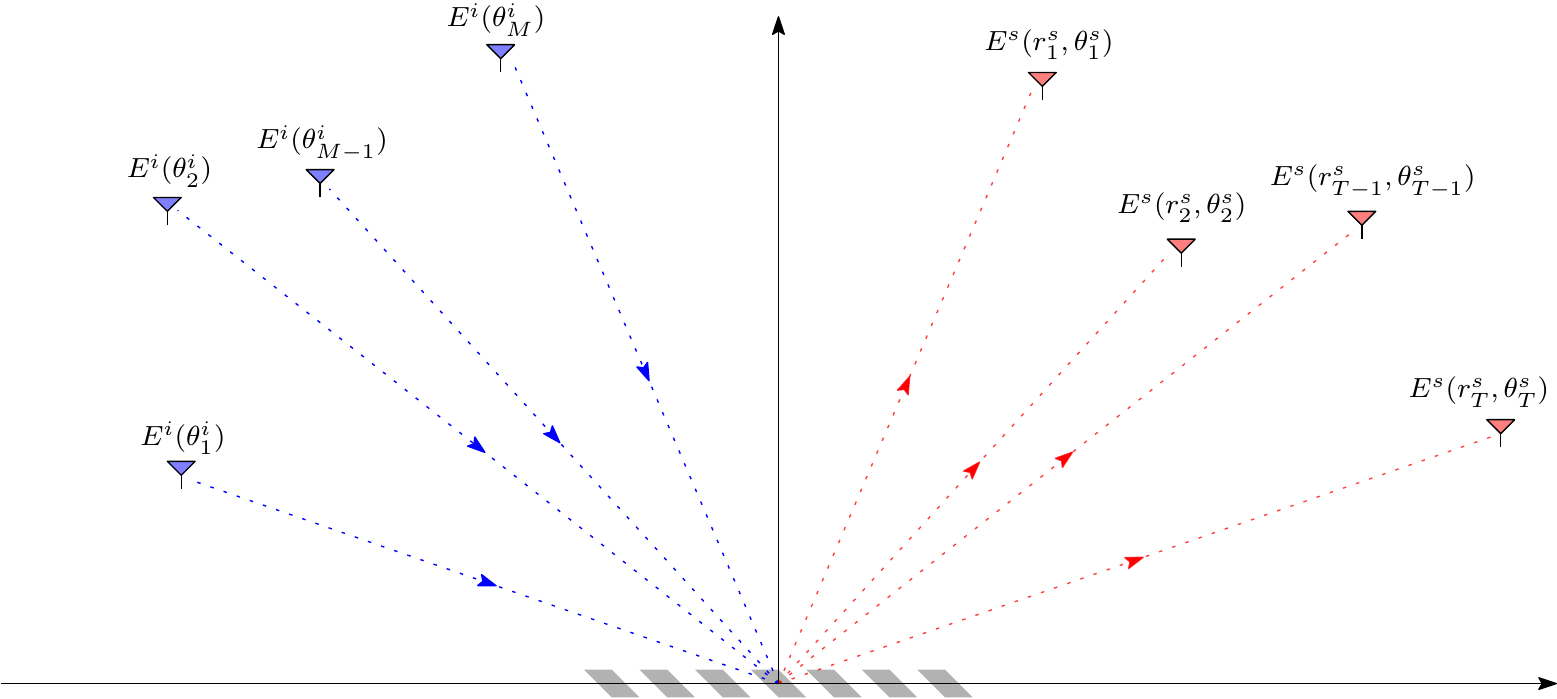}
  \caption{Multiple EM waves originating from $\theta_{1}^{i}, \ldots, \theta_{M}^{i}$ illuminate the RIS simultaneously. The scattered field are observed at mutiple points $(r_1^s, \theta_{1}^{s}), \ldots, (r_T^s, \theta_{T}^{s})$. A simple system of linear equations (\ref{E:E_s_r_theta_phi_ArrayDense}) can describe the MIMO behaviors under reasonable assumptions.}
  \label{Fig:ScatteringLinearRIS_MIMO}
\end{figure}

To shorten notations, we let
\begin{equation*}
  \begin{aligned}
    \mathbf{L} (\mathbf{r}^s) 
     & = \text{diag} \Bigl( \frac{ e^{-j 2 \pi r_1^s / \lambda }}{ r_1^s }, \frac{ e^{-j 2 \pi r_2^s / \lambda }}{ r_2^s }, \cdots, \frac{ e^{-j 2 \pi r_T^s / \lambda }}{ r_T^s } \Bigr), \\
    \textbf{cos} ( \mathbf{\Theta}^i ) 
     & = \text{diag} \bigl( \cos \theta_{1}^{i}, \cos \theta_{2}^{i}, \cdots, \cos \theta_{M}^{i} \bigr), \\
    \mathbf{W}
     & = \text{diag} \bigl( A_1 e^{j \Omega_1}, A_2 e^{j \Omega_2}, \cdots, A_N e^{j \Omega_N} \bigr).
  \end{aligned}
\end{equation*}
Given knots $x_1, \ldots, x_m$, denote by $ {\text{Vander}}_{m \times n} \bigl( x_1, \ldots, x_m \bigr)$ the Vandermonde matrix
\[
  \begin{bmatrix}
    1      & x_1    & \cdots & x_1^{n-1} \\
    \vdots & \vdots & \ddots & \vdots    \\
    1      & x_m    & \cdots & x_m^{n-1}
  \end{bmatrix} .
\]
Let 
\begin{equation*}
  \begin{aligned}
    \mathbf{V} (\mathbf{\Theta}^s) &
    = \text{Vander}_{T \times N} \bigl( e^{ j 2 \pi d \sin \theta_{1}^{s} / \lambda }, \cdots, e^{ j 2 \pi d \sin \theta_{T}^{s} / \lambda } \bigr), \\
    \mathbf{V} (\mathbf{\Theta}^i) &
    = \text{Vander}_{M \times N} \bigl( e^{ j 2 \pi d \sin \theta_{1}^{i} / \lambda }, \cdots, e^{ j 2 \pi d \sin \theta_{M}^{i} / \lambda } \bigr) .
  \end{aligned}
\end{equation*}

\begin{prop}
Suppose a linear RIS is illuminated by multiple uniform plane waves, whose strengths are $E^i (\theta^i_m), m=1, \ldots, M$. When $b_n \ll \lambda$, a canonical representation for the MIMO behaviors is given as
\begin{equation}\label{E:InputOutputEquation}
  \mathbf{E}^s (\mathbf{r}^s, \mathbf{\Theta}^s) = \frac{C}{\lambda} \mathbf{L} ( \mathbf{r}^s ) \mathbf{V}(\mathbf{\Theta}^s) \mathbf{W} {\mathbf{V} (\mathbf{\Theta}^i)}^\top \textbf{cos} ( \mathbf{\Theta}^i ) \mathbf{E}^i (\mathbf{\Theta}^i) .
\end{equation}
\end{prop}

\begin{rem}
  The proposed MIMO representation (\ref{E:InputOutputEquation}) is physically accurate and mathematically tractable compared with other popular models. The advantage is that a simple system of linear equations is employed to describe the MIMO behaviors under reasonable assumptions. It can serve as a fundamental model for analyzing and optimizing the performance of RIS-aided systems in the far-field regime.
\end{rem}

\begin{rem}
  The two Vandermonde matrices $\mathbf{V} (\mathbf{\Theta}^s)$ and $\mathbf{V} (\mathbf{\Theta}^i)$ are the steering matrices associated with the incident and scattered directions, respectively. They are well-conditioned if the knots are more or less equally spaced on unit circle \cite{pan2016bad}. The matrix $\mathbf{L} (\mathbf{r}^s)$ serves as the large-scale distance attenuation and phase delay factors. The diagonal matrix $\textbf{cos} ( \mathbf{\Theta}^i )$ is related to the particular polarization features of the EM waves impinging on RISs.
  \end{rem}

\begin{rem}
  The weight matrix $\mathbf{W}$ consists of two individual components (collecting area and phase shifting). It is worth noting that the collecting area provides us with a new degree of design freedom. The simultaneous configuration of collecting area
  and phase shifting is essential for the capability of complicated beamforming. It also converts many constrained optimization problems in the union of tori (highly non-convex) to linear space (a regular convex set) as a by-product.
\end{rem}

\section{Advantages and limitations under three typical configurations}\label{S:PerformanceAnalysis}
\noindent This section studies the advantages and limitations of RISs under three typical configurations. We first consider the random phase configuration, which avoids the overhead caused by direction-of-angle (DOA) and strength estimations. Unlike the other two configuration schemes, random phase shifting is a well-accepted benchmark  to quantify the performance enhancement \cite{liu2021reconfigurable, ding2020impact, abu2021near}.

The continuous phase compensation configuration, which is by far the most popular approach to steering the impinging waves, gives rise to the generalized Snell's law. The main disadvantage of this configuration is the anomalous mirror effect, especially when multiple incident EM waves exist.

In some applications, e.g., high-resolution sensing, RISs are required to manipulate incident EM waves to generate specific radiation patterns. We finally show the capability of complicated beamforming functions, e.g., beam focusing, endowed by simultaneous configurations of collecting area and phase shifting. Such configuration has the power to redistribute the incident waves to (almost) arbitrary radiation patterns under reasonable constraints.

It is worth pointing out that an essential difference exists between the proposed simultaneous configuration of collecting area and phase shifting and other dynamic amplitude and phase controlling schemes \cite{hu2018beyond1, puglielli2015design, long2021active, zhao2021exploiting, zhang2022holographic}. To achieve the capability of dynamic amplitude controlling with a wide range, active components, e.g., active-load impedances, are popular choices. We attempt to change amplitude by designing the geometry of unit cells \cite{liu2014broadband}.

Before proceeding, we summarize in Table~\ref{T:Summary} the capability and overhead of the proposed three typical configurations.

\begin{table}[!htbp]
\centering
\caption{Performance and overhead of the propsed three typical configurations}
  \begin{tabular}{|p{0.24\columnwidth}|p{0.30\columnwidth}|p{0.22\columnwidth}|p{0.12\columnwidth}|}
  \hline
  Configuration & Functionality \& Capability  & Estimation overhead & Beafroming complexity  \\ \hline
  Random discrete phase shifts  &  Statistical isotropic scattering patterns & \textendash  & Low\\ \hline
  Continuous phase compensation  &  Anomalous reflection \& Generalized Snell's law  & Directions of arrival and departure & Middle \\ \hline
  Simultaneous configuration of collecting area and phase shifting   &  Arbitrary beam reshaping functionality, e.g., beam splitting, beam focusing. & Full information of incident waves, e.g., directions and amplitudes,   & High \\ \hline
  \end{tabular}\label{T:Summary}
\end{table}

\subsection{Scattering via random discrete phase shifting}
\noindent The collecting area is fixed for the first configuration, but the phase shifting is discrete and randomly chosen.  The advantage lies in that RISs with random phase shifting will exhibit isotropic gain statistically. To this end, we assume $\{ \Omega_n \}$ are i.i.d. and $\text{Pr}(\Omega_n = \pi) = \text{Pr}(\Omega_n = -\pi) = 1/2$. We first analyze the \emph{statistical} SISO behaviors.

Since each unit cell re-scatters the incident EM waves independently, we see
\begin{align*}
\mathbb{E} \bigl[ E^{s} (r^s, \theta^{s}) \bigr] 
= & C \frac{ e^{-j 2 \pi r^s / \lambda } }{ r^s } E^{i} (\theta^i) \cos \theta^i \sum_{n=1}^{N} \frac{ A_n }{ \lambda } \mathbb{E} \bigl[ e^{j \Omega_n} \bigr] \text{Sa}(b_n; \theta^s; \theta^i) e^{ j 2 \pi (n-1) d \sin \theta^{i} / \lambda }  e^{ j 2 \pi (n-1) d \sin \theta^{s} / \lambda } \\
= & 0 .
\end{align*}
By a similar argument, we find
\begin{equation}\label{E:AveragedPower}
\mathbb{E} \Bigl[ \bigl| E^s (r^s, \theta^s) \bigr|^2 \Bigr] = |C|^2 \bigl( \frac{1}{r^s} \bigr)^2 \bigl( E^{i}(\theta^i) \cos \theta^i \bigr)^2 \sum_{n=1}^{N} \Bigl( \frac{A_n}{\lambda} \Bigr)^2 \text{Sa}(b_n; \theta^s; \theta^i)^2 .
\end{equation}
When $A_n = A$,
\begin{equation}\label{E:PowerRandomPhaseShifts}
\mathbb{E} \Bigl[ \bigl| E^s (r^s, \theta^s) \bigr|^2 \Bigr] = N |C|^2 \bigl( \frac{1}{ r^s } \bigr)^2 \bigl( E^{i}(\theta^i) \cos \theta^i \bigr)^2 \Bigl( \frac{A}{\lambda} \Bigr)^2 \text{Sa}(b; \theta^s; \theta^i)^2 .
\end{equation}
Moreover, when $b_n \ll \lambda$, $ \text{Sa}(b_n; \theta^s; \theta^i) \simeq 1 $. Then
\begin{equation}\label{E:PowerRandomPhaseShifts_2}
  \mathbb{E} \Bigl[ \bigl| E^s (r^s, \theta^s) \bigr|^2 \Bigr] = N |C|^2 \bigl( \frac{1}{ r^s } \bigr)^2 \bigl( E^{i}(\theta^i) \cos \theta^i \bigr)^2 \Bigl( \frac{A}{\lambda} \Bigr)^2 .
\end{equation}

For the steering and RCS functions, we see
\[
  \mathbb{E} \Bigl[ \bigl| T(\theta^s; \theta^i) \bigr|^2 \Bigr] = |C|^2 \sum_{n=1}^{N} \Bigl( \frac{A_n}{\lambda} \Bigr)^2 \text{Sa}(b_n; \theta^s; \theta^i)^2
\]
and
\begin{equation}\label{E:RCSRandomPhaseShifts}
\begin{aligned}
  \mathbb{E} \bigl[ \sigma (\theta^s; \theta^i) \bigr] 
  = 4 \pi \cos^2 \theta^i \mathbb{E} \Bigl[ \bigl| T(\theta^s; \theta^i) \bigr|^2 \Bigr] 
  = 4 \pi |C|^2 \cos^2 \theta^i \sum_{n=1}^{N} \Bigl( \frac{A_n}{\lambda} \Bigr)^2 \text{Sa}(b_n; \theta^s; \theta^i)^2.
\end{aligned}
\end{equation}
When $A_n = A$ and $b_n \ll \lambda$,
\begin{equation}\label{E:RCSRandomPhaseShifts2}
  \mathbb{E} \bigl[ \sigma (\theta^s; \theta^i) \bigr] = 4 \pi N |C|^2 \cos^2 \theta^i \Bigl( \frac{A}{\lambda} \Bigr)^2 .
\end{equation}

\begin{figure}[!htbp]
  \centering
  \includegraphics[width=0.5\columnwidth]{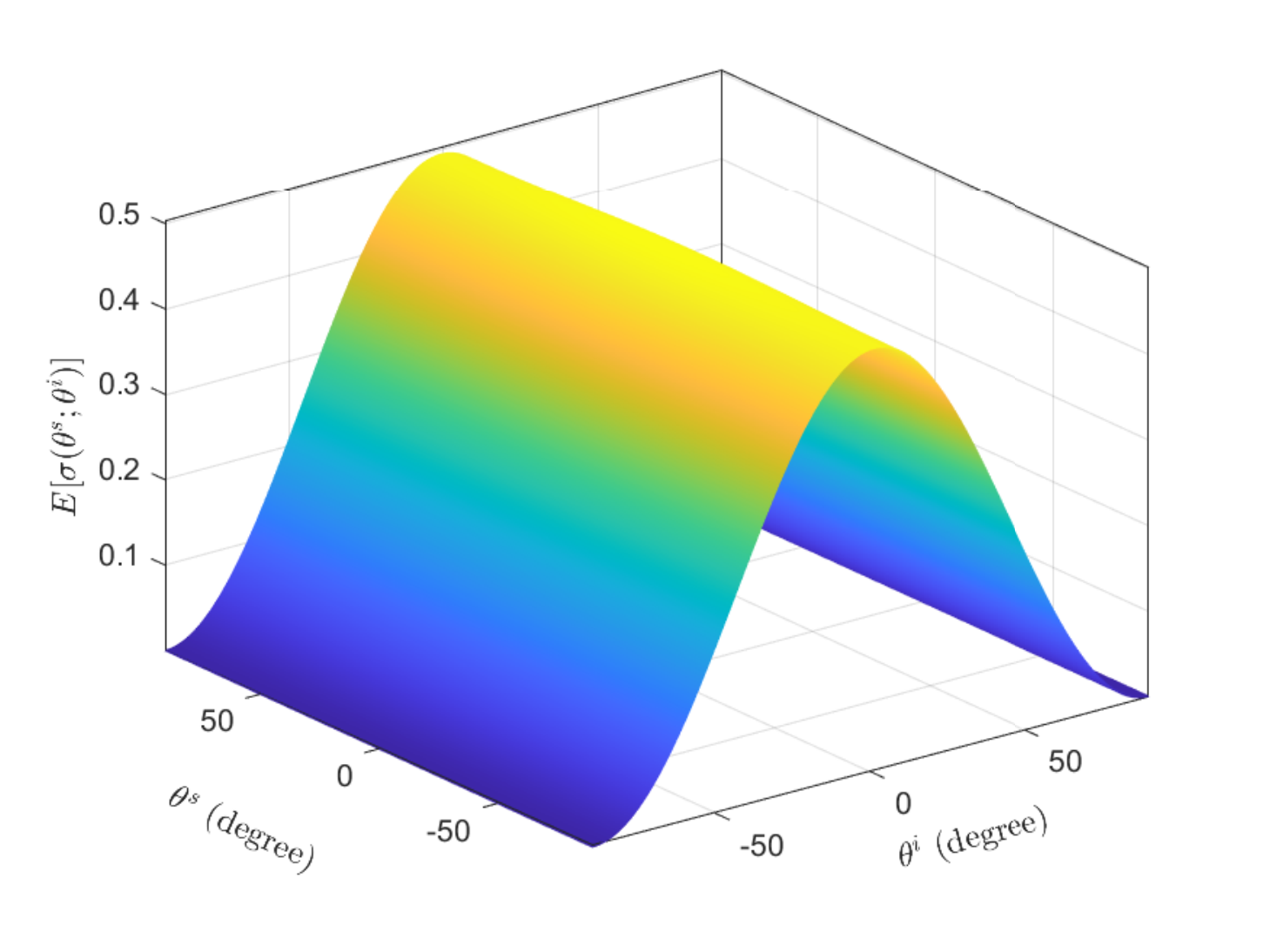}
  \caption{Statistical bistatic RCS of RISs under random phase configuration. The RIS consists of $100$ unit cells with size $0.1\lambda \times 0.1\lambda$ and spacing $d = 0.5 \lambda$. For any given incident angle $\theta^i$, RISs scatter the impinging power field equally in all directions, resulting in a (statistically) constant scattering/radiation pattern.}
  \label{Fig:RCS_RandomPhaseShifts}
\end{figure}

We now turn to the statistical MISO behaviors. For simplicity, let $\hat{\mathbf{E}}^i$ stand for ${\mathbf{V}(\mathbf{\Theta})}^\top \textbf{cos} ( \mathbf{\Theta} ) \mathbf{E}^i (\mathbf{\Theta})$. Then (\ref{E:E_s_r_theta_phi_Array}) becomes
\begin{align*}
  E^s (r^s, \theta^s) 
  = & \frac{ C }{ \lambda } \frac{ e^{-j 2 \pi r^s / \lambda }}{ r^s }
  [ 1, \cdots, e^{ j 2 \pi (N-1) d \sin \theta^s / \lambda } ]
  \mathbf{W} {\mathbf{V}(\mathbf{\Theta})}^\top \textbf{cos} ( \mathbf{\Theta} ) \mathbf{E}^i (\mathbf{\Theta}) \\
  = & \frac{ C }{ \lambda } \frac{ e^{-j 2 \pi r^s / \lambda }}{ r^s }
  [ 1, \cdots, e^{ j 2 \pi (N-1) d \sin \theta^s / \lambda } ]
  \mathbf{W} \hat{\mathbf{E}}^i .
\end{align*}
Since $\mathbb{E} \bigl[ \mathbf{W} \bigr] = \mathbf{0}$, $\mathbb{E} \bigl[ E^{s} (r^s, \theta^{s}) \bigr] = 0$. Similarly, we see
\begin{equation}\label{E:AveragedPower2}
  \begin{aligned}
    \mathbb{E} \Bigl[ \bigl| E^s (r^s, \theta^s) \bigr|^2 \Bigr] 
    = \Bigl( \frac{ | C | }{ \lambda r^s } \Bigr)^2 \mathbb{E} \Bigl[ \bigl| [ 1, \cdots , e^{ j 2 \pi (N-1) d \sin \theta^s / \lambda } ] \mathbf{W} \hat{\mathbf{E}}^i \bigr|^2 \Bigr] 
    = | C |^2 \Bigl( \frac{ 1 }{ r^s } \Bigr)^2 \sum_{n=1}^{N} \Bigl( \frac{A_n}{\lambda} \Bigr)^2  \Bigl| \hat{E}^i_n \Bigr|^2 .
  \end{aligned}
\end{equation}
When $A_n = A$,
\begin{equation}\label{E:AveragedPower3}
  \mathbb{E} \Bigl[ \bigl| E^s (r^s, \theta^s) \bigr|^2 \Bigr] = | C |^2 \Bigl( \frac{ 1 }{ r^s } \Bigr)^2 \Bigl( \frac{A}{\lambda} \Bigr)^2 \bigl\Vert \hat{\mathbf{E}}^i \bigr\Vert^2.
\end{equation}

We conclude through (\ref{E:AveragedPower})-(\ref{E:RCSRandomPhaseShifts2}) that, under the random phase configuration, RISs scatter the incident power field equally in all directions giving a (statistically) constant scattering/radiation pattern, as illustrated in Fig.~\ref{Fig:RCS_RandomPhaseShifts}. We conclude through (\ref{E:AveragedPower2}) and (\ref{E:AveragedPower3}) that the RIS has a statistical isotropic scattering pattern even when multiple incident EM waves exist.

We thus refer to RISs with random phase configuration as statistical isotropic RISs. Although the random phase configuration cannot effectively use the available degrees of freedom, it is a well-accepted benchmark for quantifying the other two configurations.

\subsection{Anomalous reflection via phase compensations}
\noindent Phase compensation configuration is the most popular approach approach to steer an incident EM wave to a preferred reflective direction. We now address the advantages and limitations. We say linear RISs work in phase compensation configuration mode if the phase shifting of each unit cell can be written as $\Omega_n = - 2 \pi (n-1) d \Delta / \lambda + 2 k \pi$, $n = 1, \ldots, N$.

Consider SISO behaviors first. Suppose a linear RIS is illuminated by one uniform plane wave originating from $\theta^{i}$. Then, with (\ref{E:E_s_r_theta_phi_Array}), the scattered field at $(r^s, \theta^{s})$ is given as
\begin{equation*}
  \begin{aligned}
    E^{s} (r^s, \theta^{s})
    = & C \frac{ e^{-j 2 \pi r^s / \lambda }}{ r^s } E^{i}(\theta^i) \cos \theta^i \sum_{n=1}^{N} \frac{ A_n }{ \lambda } e^{j \Omega_n} e^{ j 2 \pi (n-1) d (\sin \theta^{i} + \sin \theta^{s}) / \lambda } \\
    = & C \frac{ e^{-j 2 \pi r^s / \lambda }}{ r^s } E^{i}(\theta^i) \cos \theta^i \sum_{n=1}^{N} \frac{ A_n }{ \lambda } e^{ j 2 \pi (n-1) d (\sin \theta^{i} + \sin \theta^{s} - \Delta) / \lambda } .
  \end{aligned}
\end{equation*}
It is easily seen that
\begin{equation}\label{E:ScatteringAbs}
  \bigl| E^{s} (r^s, \theta^{s}) \bigr| \le \frac{|C|}{ r^s } \bigl( E^{i}(\theta^i) \cos \theta^i \bigr) \sum_{n=1}^{N} \frac{ A_n }{ \lambda }.
\end{equation}

For any given $\theta^i$, if there exists a $\theta^s$ such that $\Delta = \sin \theta^{i} + \sin \theta^{s}$, equality holds in (\ref{E:ScatteringAbs}). The scattered field achieves its global maximum at $\theta^s$. That means the RIS is configured to reflect the incident EM wave originating from $\theta^i$ to $\theta^s$. If all unit cells have the same collecting area $A$, the received power at $(r^s, \theta^s)$ is given as
\begin{equation}\label{E:PowerPhaseCompensation}
  \bigl| E^{s} (r^s, \theta^{s}) \bigr|^2 = N^2 | C |^2 \bigl( \frac{ 1 }{ r^s } \bigr)^2 \bigl( E^{i}(\theta^i) \cos \theta^i \bigr)^2 \Bigl( \frac{ A }{ \lambda } \Bigr)^2.
\end{equation}

\begin{rem}
  The proposed equation (\ref{E:PowerPhaseCompensation}) demonstrates several desired properties with phase compensation configurations. The most important one is the inter-unit cooperative beamforming gain. This highlights the fact that the maximum power reflected by RISs could grow quadratically with the total collecting area of all unit cells (not the physical aperture of RIS). Therefore, large RISs could be leveraged to compensate for the significant power loss over the two paths via beamforming. Comparison of (\ref{E:PowerPhaseCompensation}) and (\ref{E:PowerRandomPhaseShifts}) shows that the power gain brought by phase compensations becomes $N$ times greater than random phase shifting.
\end{rem}

\subsubsection{Grating lobes}
\noindent It is well known that, when $d / \lambda > 1/2$, grating lobes will occur in the visible region for some incident angles. In fact, a necessary and sufficient condition for a grating lobe at $\widetilde{\theta^s} \neq \theta^s$ is that there exists two integers $k \neq 0$ and $k'$ such that
\begin{equation}\label{E:SufficientConditionGratingLobes}
  d ( \sin \theta^{i} + \sin \theta^{s} - \Delta ) / \lambda = d ( \sin \theta^{i} + \sin \widetilde{\theta^{s}} - \Delta ) / \lambda + k = k'.
\end{equation}
A necessary condition for (\ref{E:SufficientConditionGratingLobes}) is $d ( \sin \theta^{s} - \sin \widetilde{\theta^{s}} ) / \lambda = k$. To ensure the existence of such a grating lobe at $\widetilde{\theta^{s}}$, we must have $d / \lambda > 1/2$.

We demonstrate the scattered field for various values of $d / \lambda$ in Fig.~\ref{Fig:RIS_Aliasing}. It illustrates the concept of grating lobes. The RIS, consisting of $100$ unit cells with size $0.1\lambda \times 0.1\lambda$, is configured to steer the incident wave originating from $30^\circ$ to $-50^\circ$ by setting $\Delta = \sin(30^\circ) + \sin(-50^\circ)$. The incident wave is plotted in blue. The scattered electric fields at $r^s = 100 \lambda$ are plotted in red, from which we can see that the main lobes are centered around $-50^\circ$. For the case of $d / \lambda = 0.7$ (the solid red curve), there does exist a grating lobe centered around $41.49^\circ$.

\begin{figure}[!htbp]
  \centering
  \includegraphics[width=0.46\columnwidth]{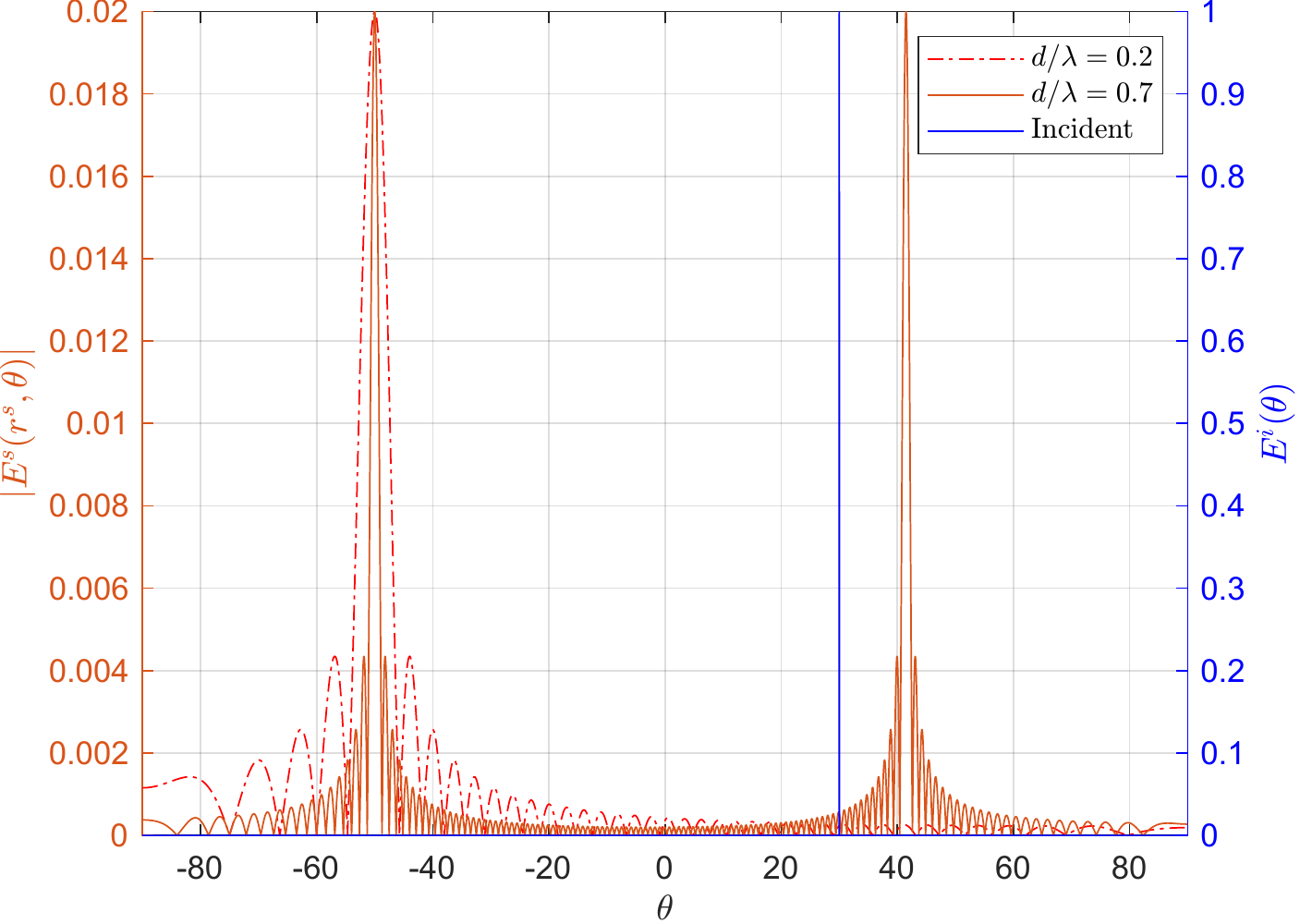}
  \caption{Grating lobe. The RIS is configured to steer the incident wave from $30^\circ$ to $-50^\circ$. The spike plotted in blue is the EM wave impinging upon RIS. The scattered fields at $r^s = 100 \lambda$ are plotted in red, from which we can see that the main lobes are both centered around $-50^\circ$. For the case of $d / \lambda = 0.7$ (the solid red curve), there does exist a grating lobe centered around $41.49^\circ$.}
  \label{Fig:RIS_Aliasing}
\end{figure}

\subsubsection{Anomalous reflection}
\noindent A significant limitation of beam steering through continuous phase compensation configurations is the anomalous mirror effect.

Specifically, for a linear RIS, we can choose $\Delta = \sin \theta^{i} + \sin \theta^{s}$ such that the incident wave from $\theta^i$ is steered to $\theta^s$, i.e., $d ( \sin \theta^{i} + \sin \theta^{s} - \Delta ) / \lambda = 0$. It is easy to check that, for any $| \Delta | < 2$, there are infinitely many pairs $( \widetilde{\theta^i}; \widetilde{\theta^s} ) \neq ( \theta^i; \theta^s )$ satisfying $d ( \sin \widetilde{\theta^{i}} + \sin \widetilde{\theta^{s}} - \Delta ) / \lambda + k = 0$, where $k$ is an integer. That means the RIS also steers the incident wave from $\widetilde{ \theta^i }$ to $\widetilde{ \theta^{s} }$. We theoretically demonstrate  the generalized Snell's law of reflections (anomalous mirror effect).

We illustrate the anomalous mirror effect in Fig.~\ref{Fig:Weight_TimeDelay}. There exist two EM waves impinging upon the RIS. It is configured to steer the wave from $30^\circ$ to $-50^\circ$. From the scattered field observed at $r^s = 100 \lambda$, we find two main lobes exist in the angular range. The higher one centered around $50^\circ$ is introduced by the incident wave from $30^\circ$. The lower one centered around $\theta = 52.59^\circ$ is related to the incident wave from $70^\circ$. RIS operates as an anomalous mirror.

\begin{figure}
  \centering
  \subfigure[]{
    \label{Fig:Weight_TimeDelay}
    \includegraphics[width=0.46\columnwidth]{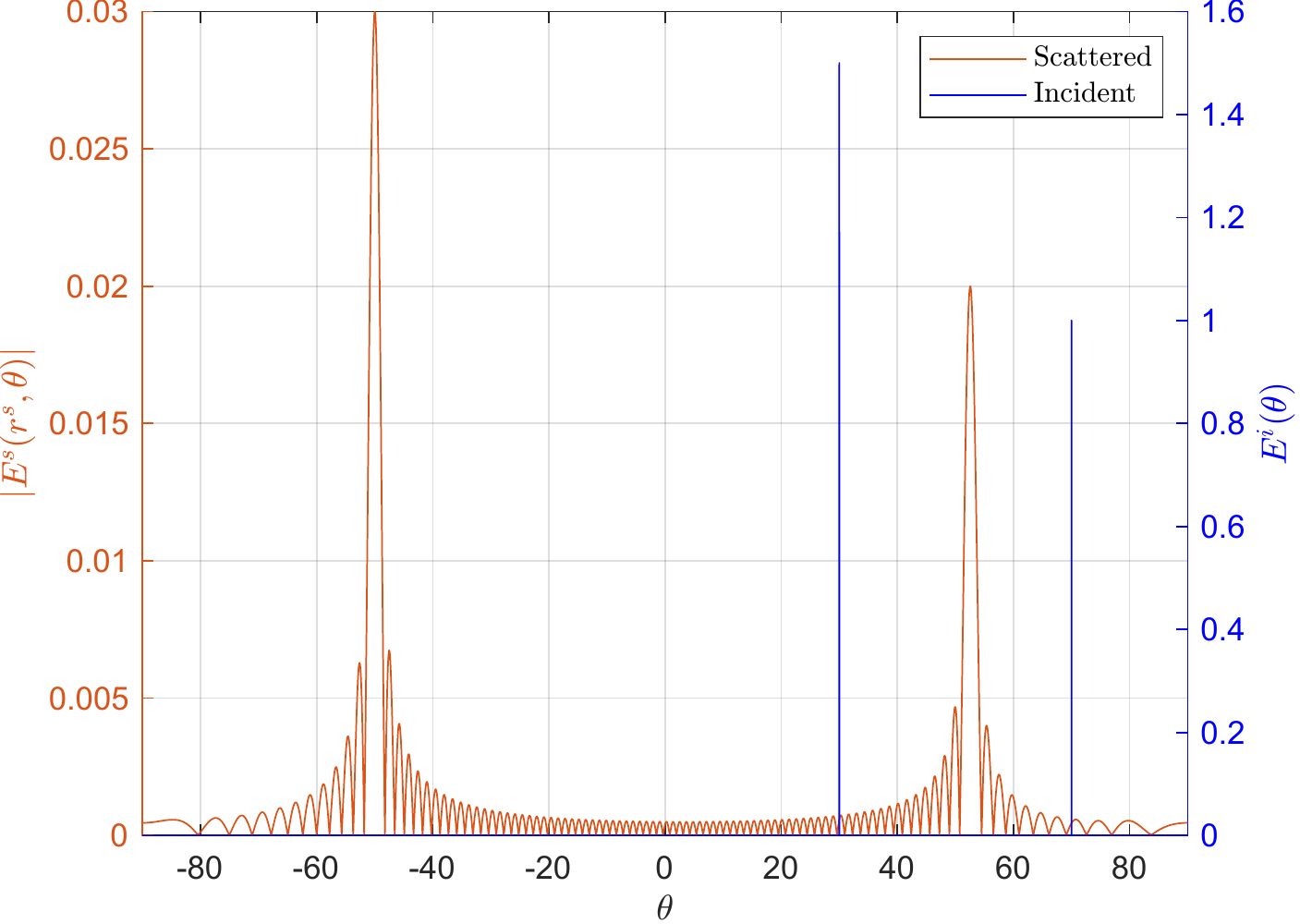}}
  \subfigure[]{
    \label{Fig:Weight_Fourier}
    \includegraphics[width=0.46\columnwidth]{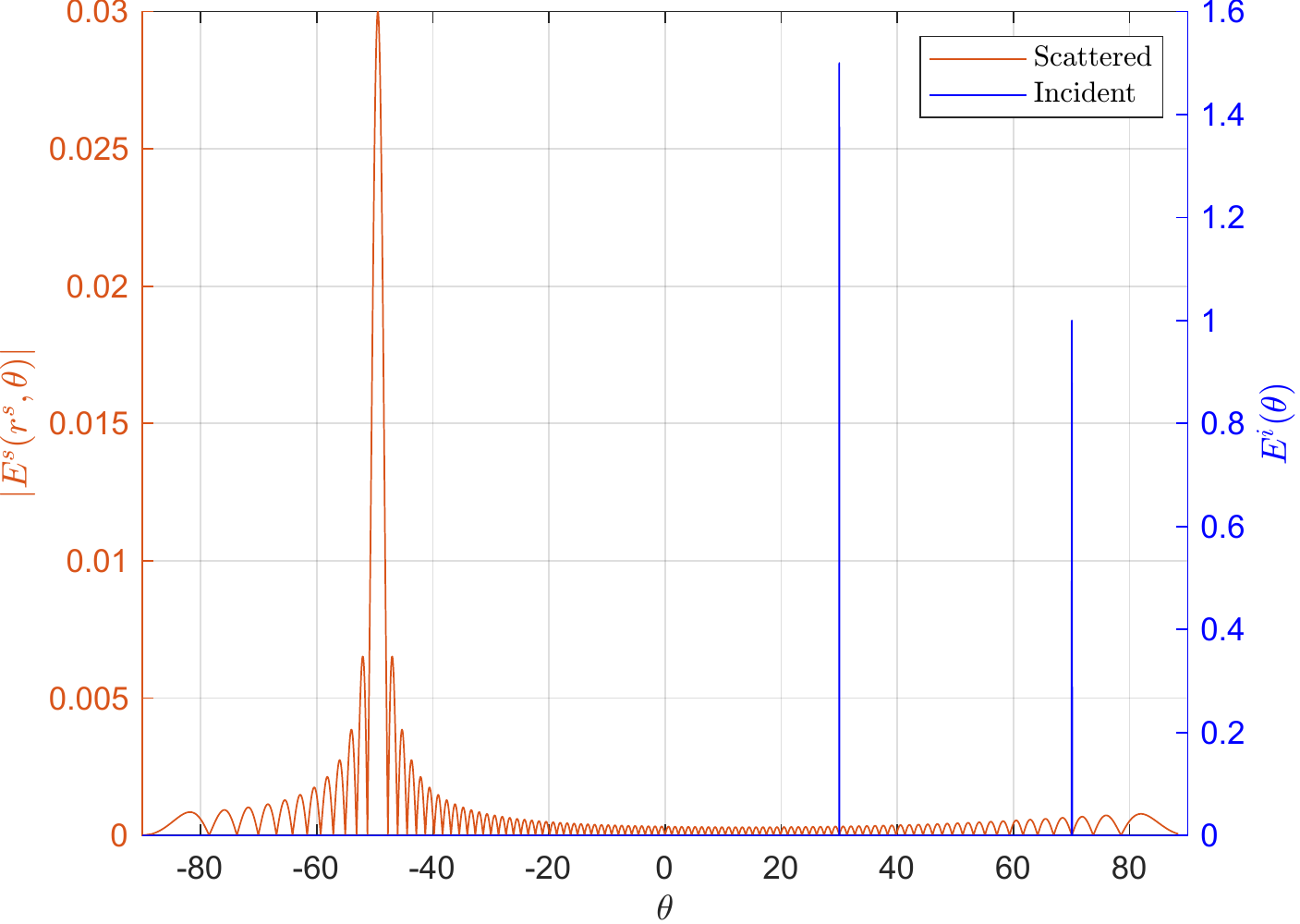}}
  \caption{(a) Anomalous reflections. The RIS is configured to steer the one originating from $30^\circ$ to $-50^\circ$. There are two main lobes in the angular range (in red). The higher one centered around $50^\circ$ is introduced by the incident wave from $30^\circ$. The lower one around $\theta = 52.59^\circ$ is related to the incident wave from $70^\circ$. (b) Arbitrary beam reshaping. The RIS is configured to perform energy focusing. The main lobe centered around $\theta = 52.59^\circ$ in (a) is removed.}
  \label{Fig:Weight_TimeDelay_Fourier}
\end{figure}

To have a better understanding of the difference between specular and anomalous reflections, we turn to the bistatic RCS functions. It is easy to check that the steering function is
\begin{equation}\label{E:Steer_PhaseCompensation}
  \begin{aligned}
    T (\theta^i, \theta^s) = C \sum_{n=1}^{N} \frac{ A_n }{ \lambda } e^{j \Omega_n} e^{ j 2 \pi (n-1) d \sin \theta^{i} / \lambda }  e^{ j 2 \pi (n-1) d \sin \theta^{s} / \lambda } 
    = C \sum_{n=1}^{N} \frac{ A_n }{ \lambda } e^{ j 2 \pi (n-1) d (\sin \theta^{i} + \sin \theta^{s} - \Delta) / \lambda } .
  \end{aligned}
\end{equation}
The bistatic RCS function is then given by
\begin{equation}\label{E:RCS_PhaseCompensation}
  \begin{aligned}
    \sigma (\theta^i, \theta^s) = 4 \pi \cos^2 \theta^i \bigl| T(\theta^s; \theta^i) \bigr|^2  
    = 4 \pi |C|^2 \cos^2 \theta^i \biggl| \sum_{n=1}^{N} \frac{ A_n }{ \lambda } e^{ j 2 \pi (n-1) d (\sin \theta^{i} + \sin \theta^{s} - \Delta) / \lambda } \biggr|^2 .
  \end{aligned}
\end{equation}

We plot the steering and bistatic RCS functions in Figs.~\ref{F:Transfer_RCS_Reflection} and \ref{F:Transfer_RCS_PhaseCompensation} to compare the difference between specular and anomalous reflections. We first set $\Delta = 0$ to make the RIS operate as a specular mirror. The steering and RCS functions are plotted in Fig.~\ref{F:Transfer_RCS_Reflection}. It is then configured to steer the incident wave from $30^\circ$ to $-50^\circ$ by setting $\Delta = \sin(30^\circ) + \sin(-50^\circ)$ as plotted in Fig.~\ref{F:Transfer_RCS_PhaseCompensation}. A comparison of Figs.~\ref{F:Transfer_RCS_Reflection} and \ref{F:Transfer_RCS_PhaseCompensation} shows that the steering and RCS functions of anomalous reflection can be considered as non-linear translations of the corresponding two functions of specular reflection.

\begin{figure}[!htbp]
  \centering
  \subfigure[]{
    \label{F:Transfer_RCS_Reflection:1}
    \includegraphics[width=0.45\columnwidth]{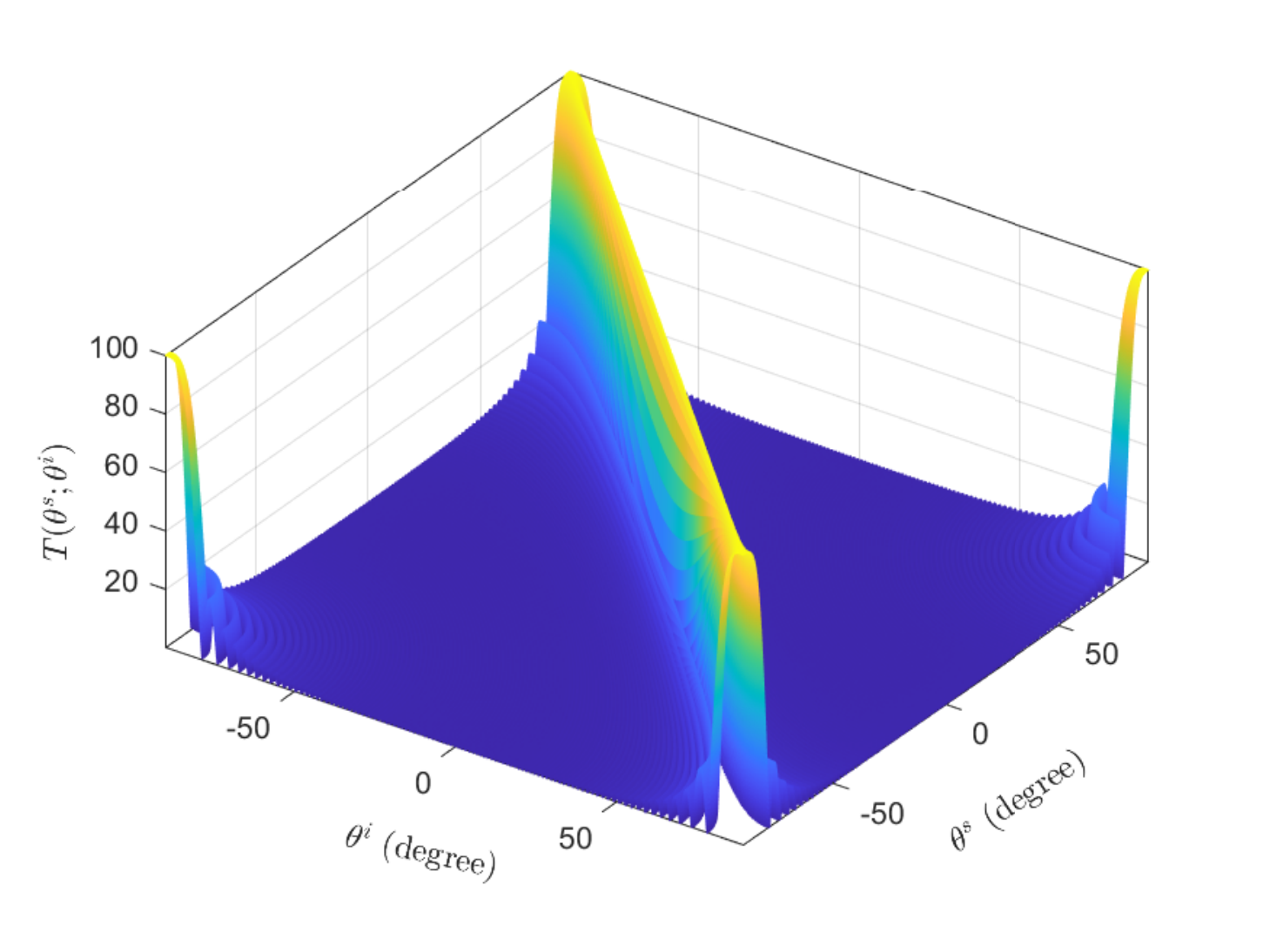}}
  \subfigure[]{
    \label{F:Transfer_RCS_Reflection:2}
    \includegraphics[width=0.45\columnwidth]{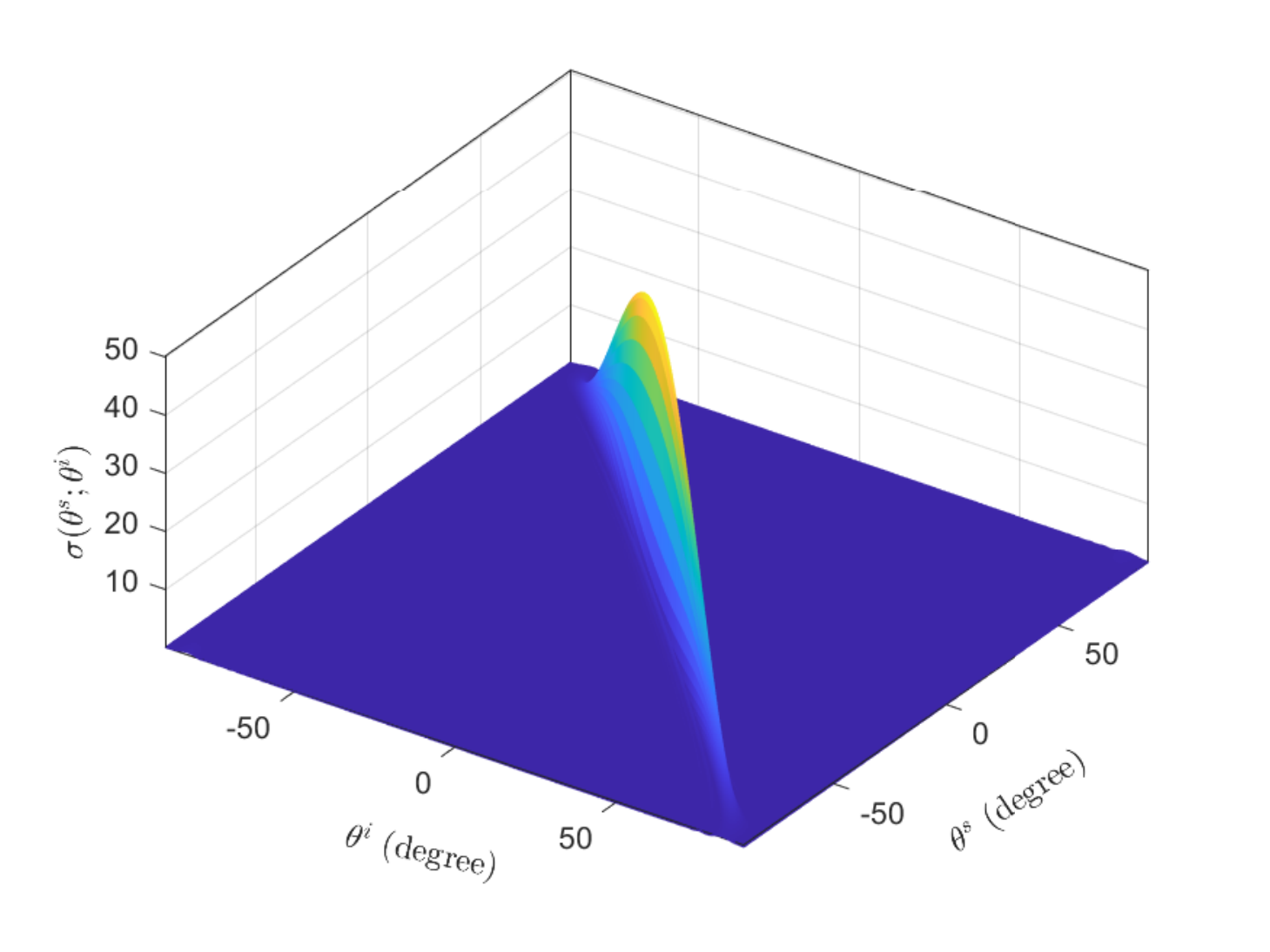}}
  \caption{The steering and bistatic RCS functions of specular reflection. We set $\Delta = 0$ to make the RIS operate as a specular mirror. (a) Steering function. (b) Bistatic RCS function.}
  \label{F:Transfer_RCS_Reflection}
\end{figure}

\begin{figure}[!htbp]
  \centering
  \subfigure[]{
    \label{F:Transfer_RCS_PhaseCompensation:1}
    \includegraphics[width=0.45\columnwidth]{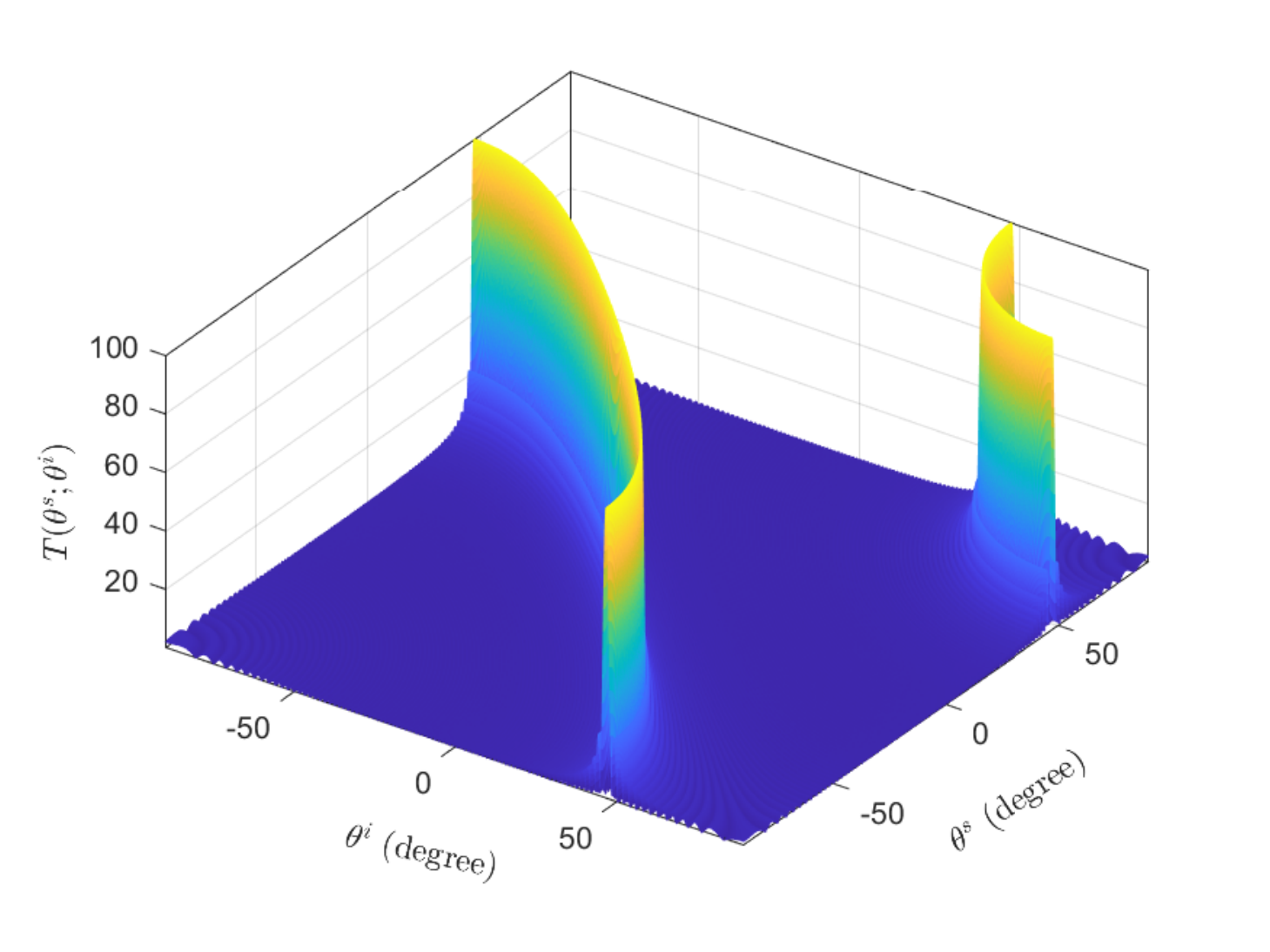}}
  \subfigure[]{
    \label{F:Transfer_RCS_PhaseCompensation:2}
    \includegraphics[width=0.45\columnwidth]{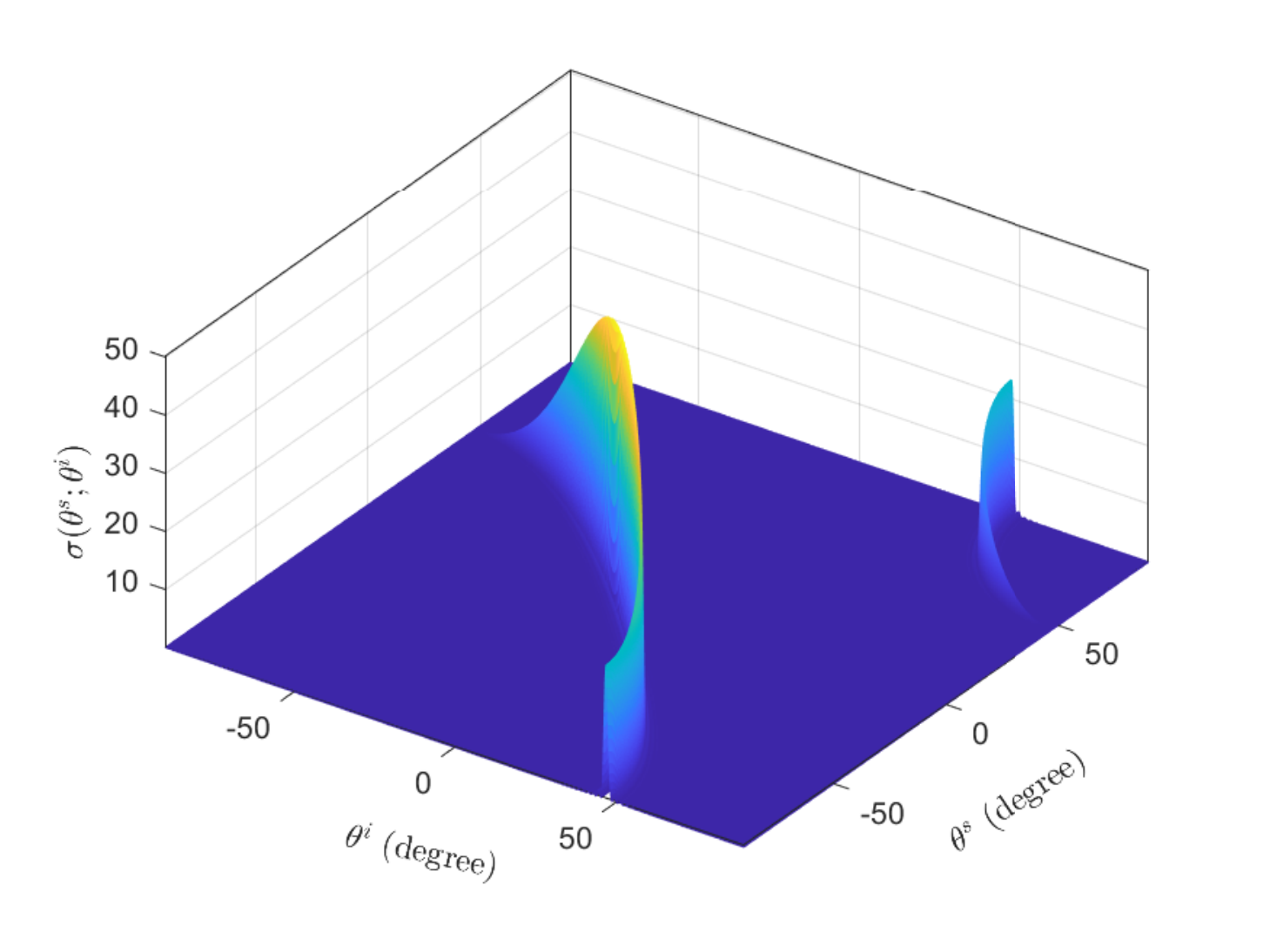}}
  \caption{The steering and bistatic RCS functions of anomalous reflection. The RIS is configured to steer the incident wave from $30^\circ$ to $-50^\circ$. The steering and RCS functions are non-linear translations of the corresponding two functions of
  specular reflection. (a) Steering function. (b) Bistatic RCS function. }
  \label{F:Transfer_RCS_PhaseCompensation}
\end{figure}

\subsection{Arbitrary beam reshaping}
\noindent We now discuss how to redistribute the incident power to produce an arbitrary (desired) power radiation pattern. The problem could be formulated as, given an incident field $ \mathbf{E}^i (\mathbf{\Theta}^i)$, how to choose $\mathbf{W}$ so that the actual scattered field $\mathbf{E}^s (r^s, \mathbf{\Theta}^s)$ approximates the desired field $\mathbf{E}^s_d (r^s, \mathbf{\Theta}^s)$, i.e., 
\begin{equation*}
  \underbrace{\mathbf{E}^s_d (r^s, \mathbf{\Theta}^s)}_{ \text{Desired} } \approx \mathbf{E}^s (r^s, \mathbf{\Theta}^s) 
  = \frac{C}{\lambda} \frac{ e^{-j 2 \pi r_s / \lambda }}{ r_s } \mathbf{V}(\mathbf{\Theta}^s) \overbrace{\mathbf{W}}^{ \text{Configuration} } \underbrace{ {\mathbf{V} (\mathbf{\Theta}^i)}^\top \textbf{cos} ( \mathbf{\Theta}^i ) \mathbf{E}^i (\mathbf{\Theta}^i) }_{ = \hat{\mathbf{E}}^i,\ \text{Known}}.
\end{equation*}

A most straightforward approach is to minimize the $\ell_2$-norm of the difference between desired and actual output, i.e.,
\begin{equation}\label{E:ScatteredFieldApproximation1}
  \mathbf{W} = {\arg\min}_{\mathbf{W}} \bigl\Vert \mathbf{E}^s (r^s, \mathbf{\Theta}^s) - \mathbf{E}^s_d (r^s, \mathbf{\Theta}^s) \bigr\Vert .
\end{equation}
For simplicity of notations, we let $\hat{\mathbf{E}}^i = {\mathbf{V}( \mathbf{\Theta}^i )}^\top \textbf{cos} ( \mathbf{\Theta}^i ) \mathbf{E}^i (\mathbf{\Theta}^i)$ and denote by $\beta$ the constant $\frac{ N C }{ \lambda } \frac{ e^{-j 2 \pi r_s / \lambda }}{ r_s }$. It is easily seen that (\ref{E:ScatteredFieldApproximation1}) is equivalent to
\begin{equation}\label{E:ScatteredFieldApproximation2}
  \min_{\mathbf{W}} \bigl\Vert \frac{\beta}{N} \mathbf{V}( \mathbf{\Theta}^s ) \mathbf{W}  \hat{\mathbf{E}}^i - \mathbf{E}^s_d (r^s, \mathbf{\Theta}^s) \bigr\Vert .
\end{equation}

To illustrate the potential of the simultaneous configuration of the collecting area and phase shifting, we consider a particular situation where the spacing $d = \lambda / 2$ and the scattered field are sampled at regular grids $\theta^s_n = \arcsin (-1 + 2 (n-1) / N)$, $n = 1, \ldots, N$. The advantage lies in the fact that the Vandermonde matrix $\mathbf{V}( \mathbf{\Theta}^s )$ is perfectly conditioned because the knots are equally spaced on the unit circle \cite{pan2016bad}. In fact, $ \mathbf{V}( \mathbf{\Theta}^s ) / \sqrt{N} $ is an angular discrete Fourier transform (DFT) matrix, i.e., $\mathbf{V}( \mathbf{\Theta}^s ) \mathbf{V} ( \mathbf{\Theta}^s )^{H} = \mathbf{V} ( \mathbf{\Theta}^s )^{H} \mathbf{V}( \mathbf{\Theta}^s ) = N \mathbf{I}$.

If we further write $\hat{\mathbf{E}}^s_d = \mathbf{V} ( \mathbf{\Theta}^s )^{H} \mathbf{E}^s_d (r^s, \mathbf{\Theta}^s)$, (\ref{E:ScatteredFieldApproximation2}) is equivalent to
\[
  \min_{\mathbf{W}} \bigl\Vert \beta \mathbf{W} \hat{\mathbf{E}}^i - \hat{\mathbf{E}}^s_d \bigr\Vert .
\]
A solution is given as
\[
  W_n = A_n e^{j \Omega_n} =
  \left\{
  \begin{array}{ll}
    \beta^{-1} \hat{E}^s_{d,n} / \hat{E}^i_n , & \hbox{$\hat{E}^i_n \neq 0$;} \\
    0,                                         & \hbox{$\hat{E}^i_n = 0$.}
  \end{array}
  \right.
\]

For generic scattered steering matrix $\mathbf{V}( \mathbf{\Theta}^s )$, provided $ \mathbf{V} ( \mathbf{\Theta}^s )^{H} \mathbf{V} ( \mathbf{\Theta}^s ) $ is invertible (i.e., $\mathbf{V} ( \mathbf{\Theta}^s )$ is a full column rank matrix), we let $\hat{\mathbf{E}}^s_d = \bigl( \mathbf{V} ( \mathbf{\Theta}^s )^{H} \mathbf{V} ( \mathbf{\Theta}^s ) \bigr)^{-1} \mathbf{V} ( \mathbf{\Theta}^s )^{H} \mathbf{E}^s_d (r^s, \mathbf{\Theta}^s)$. A feasible solution to (\ref{E:ScatteredFieldApproximation2}) is given as
\begin{equation}\label{E:WeightAreaPhase}
  W_n = A_n e^{j \Omega_n} =
  \left\{
  \begin{array}{ll}
    \beta^{-1} \hat{E}^s_{d,n} / \hat{E}^i_n, & \hbox{$\hat{E}^i_n \neq 0$;} \\
    0,                                        & \hbox{$\hat{E}^i_n = 0$.}
  \end{array}
  \right.
\end{equation}

\begin{rem}
If $\mathbf{V} ( \mathbf{\Theta}^s )$ is ill-conditioned or some $\hat{E}^i_n$ is close to zero, the scattered field $\mathbf{E}^s (r^s, \mathbf{\Theta}^s)$ can approximate the desired $\mathbf{E}^s_d (r^s, \mathbf{\Theta}^s)$, but $\hat{\mathbf{E}}^s_d$ or $\mathbf{W}$ will have a large norm. This means that the RIS with this configuration is more sensitive to perturbations. Low rank approximates of $\mathbf{V} ( \mathbf{\Theta}^s )$ or the truncated version of $\hat{\mathbf{E}}^i$ should be used. This ensures that the norms of $\hat{\mathbf{E}}^s_d$ and $\mathbf{W}$ will not be unnecessarily large.
\end{rem}

We now employ an example to demonstrate the arbitrary beam reshaping capability endowed by the simultaneous configuration of the collecting area and phase shifting. To make a fair comparison with phase compensation as illustrated in Fig.~\ref{Fig:Weight_TimeDelay}, the RIS is expected to focus two incident waves from $30^\circ$ and $70^\circ$ to produce a single main lobe centered around $-50^\circ$. The configuration is given through (\ref{E:WeightAreaPhase}). The final scattered electric field at $r^s = 100 \lambda$ is plotted in Fig.~\ref{Fig:Weight_Fourier}, from which we find the main lobe centered around $\theta = 52.59^\circ$ in Fig.~\ref{Fig:Weight_TimeDelay} is removed. We thus demonstrate the functionality of energy focusing through the simultaneous configuration of collecting area and phase shifting.

\section{Conclusion}
\noindent In this paper, we have employed continuous and discrete strategies to model a single patch and patch array and their interactions with multiple incident EM waves. We introduced a physically accurate formula to calculate the electric field scattered by a rectangular metallic patch. Several mathematically tractable models were proposed to characterize the input/output behaviors of patch array. Particularly, a simple system of linear equations was used to describe the MIMO behavior of linear RISs. Based on the proposed models, we have evaluated the performance of RISs under three typical configurations. We proved that a good solution to complicated beam reshaping functionality is through the simultaneous configuration of collecting area and phase shifting.

\appendices

\section{Proof of Theorem~\ref{T:ScatterdFieldPatch}}\label{S:Appendix_A}
\begin{proof}
  We follow similar steps as in \cite[Ch.~6.8]{balanis2012advanced}. Denote by $\mathbf{E}^i$ the incident electric field on the plate. The electric and magnetic fields could be written in Cartesian coordinates as
  \[
    \mathbf{E}^i = E^{i} \bigl( -\mathbf{e}_x \sin \phi^i  + \mathbf{e}_y \cos \phi^i \bigr) e^{ - j 2 \pi (-x \sin \theta^i \cos \phi^i - y \sin \theta^i \sin \phi^i - z \cos \theta^i ) / \lambda}
  \]
  and
  \[
    \mathbf{H}^i = \frac{ E^{i} }{ \eta } \bigl( \mathbf{e}_x \cos \theta^i \cos \phi^i + \mathbf{e}_y \cos \theta^i \sin \phi^i - \mathbf{e}_z \sin \theta^i \bigr) e^{ - j 2 \pi ( - x \sin \theta^i \cos \phi^i - y \sin \theta^i \sin \phi^i - z \cos \theta^i ) / \lambda} .
  \]
  Here $\eta$ is the intrinsic impedance of the medium.

  The starting point is the electric current density (measured in $A / m^2$) induced on the plate, which is given as 
  \[
    \mathbf{J}^s = \mathbf{n} \times \mathbf{H}^{\text{total}} |_{z=0} = \mathbf{e}_z \times \bigl ( \mathbf{H}^{i} + \mathbf{H}^{r} \bigr) |_{z=0}.
  \]
  The reflected magnetic field $\mathbf{H}^{r}$ could be represented by $\mathbf{H}^{r} = - \Gamma \mathbf{H}^{i}$, where $\Gamma$ is the reflection coefficient of the surface. We then write $\mathbf{J}^s$ in Cartesian coordinates as
  \begin{align*}
    \mathbf{J}^s 
     = & \mathbf{e}_z \times \bigl( \mathbf{H}^{i} + \mathbf{H}^{r} \bigr) |_{z=0} =  \mathbf{e}_z \times (1-\Gamma) \mathbf{H}^{i} |_{z=0} \\
     = & (1-\Gamma) \frac{ E^{i} }{\eta} \bigl( - \mathbf{e}_x \cos \theta^i\sin\phi^i + \mathbf{e}_y \cos \theta^i\cos \phi^i \bigr) e^{ - j 2 \pi (- x \sin \theta^i \cos \phi^i - y \sin \theta^i \sin \phi^i ) / \lambda } .
  \end{align*}
  Or equivalently,
  \begin{equation*}
    J^s_x =
    \begin{dcases*}
      - (1-\Gamma) \frac{E^{i} }{\eta} \cos \theta^i \sin \phi^i e^{ - j 2 \pi (-x \sin \theta^i \cos \phi^i - y \sin \theta^i \sin \phi^i ) / \lambda }, & for $ (x,y) \in S$; \\
      0,      & for elsewhere,
    \end{dcases*}
  \end{equation*}
  \begin{equation*}
    J^s_y =
    \begin{dcases*}
      (1-\Gamma) \frac{E^{i} }{\eta} \cos \theta^i \cos \phi^i e^{ - j 2 \pi (-x \sin \theta^i \cos \phi^i - y \sin \theta^i \sin \phi^i ) / \lambda }, & for $(x,y) \in S$; \\
      0,      & for elsewhere,
    \end{dcases*}
  \end{equation*}
  and
  \begin{equation*}
    J^s_z = 0, \quad \hbox{for everywhere}.
  \end{equation*}

  The structure of unit cells is designed to change surface current density $\mathbf{J}^s$ properly. We emphasize that the reflection coefficient $\Gamma$ is determined by the surface impedance. For a perfect electric conductor, the magnetic current density (measured in $V / m^2$) induced on the plate is $\mathbf{M}^s = \mathbf{0}$, i.e., $M^s_x = M^s_y = M^s_z = 0$, for everywhere.
  
  Next we will evaluate \cite[eq.~(6-123)]{balanis2012advanced}
  \begin{equation*}
  \begin{aligned}
    \mathbf{N}^s 
    = & \iint_{S} \mathbf{J}^s e^{j 2 \pi r' cos \psi / \lambda } d \Omega 
    = \iint_{S} \bigl( \mathbf{e}_x J^s_x + \mathbf{e}_y J^s_y + \mathbf{e}_z J^s_z \bigr) e^{j 2 \pi r' cos \psi / \lambda } d \Omega , \\
    \mathbf{L}^s
    = & \iint_{S} \mathbf{M}^s e^{j 2 \pi r' cos \psi / \lambda } d \Omega 
    = \iint_{S} \bigl( \mathbf{e}_x M^s_x + \mathbf{e}_y M^s_y + \mathbf{e}_z M^s_z \bigr) e^{j 2 \pi r' cos \psi / \lambda } d \Omega .
  \end{aligned}
  \end{equation*}
  The integrations are performed over the entire region occupied by the plate, as illustrated in Fig.~\ref{Fig:Scattering_OneIncident}. Since $\mathbf{M}^s = \mathbf{0}$, it is obvious that $\mathbf{L}^s = \mathbf{0}$, i.e., $L^s_{\theta} = L^s_{\phi} = L^s_{r} = 0$.
  
  We now turn to the evaluation of $\mathbf{N}^s$. With the rectangular-to-spherical component transformation
  \[
    \begin{bmatrix}
      \mathbf{e}_{r}      \\
      \mathbf{e}_{\theta} \\
      \mathbf{e}_{\phi}
    \end{bmatrix}
    =
    \begin{bmatrix}
      \sin \theta \cos \phi & \sin \theta \sin \phi & \cos \theta   \\
      \cos \theta \cos \phi & \cos \theta \sin \phi & - \sin \theta \\
      -\sin \phi            & \cos \phi             & 0
    \end{bmatrix}
    \begin{bmatrix}
      \mathbf{e}_x \\
      \mathbf{e}_y \\
      \mathbf{e}_z
    \end{bmatrix} ,
  \]
  we see that
  \begin{equation*}
    \begin{dcases*}
      \begin{aligned}
      N^s_{\theta} = & \iint_{S} \bigl[ J^s_x \cos \theta^s \cos \phi^s + J^s_y \cos \theta^s \sin \phi^s - J^s_z \sin \theta^s \bigl] e^{j 2 \pi r' \cos \psi / \lambda } d \Omega , \\
      N^s_{\phi} = & \iint_{S} \bigl[ - J^s_x \sin \phi^s + J^s_y \cos \phi^s \bigl] e^{j 2 \pi r' \cos \psi / \lambda } d \Omega .
      \end{aligned}
    \end{dcases*}
  \end{equation*}
  Using the fact $r' \cos \psi = \mathbf{r}' \cdot \mathbf{e}_r = x \sin \theta^s \cos \phi^s + y \sin \theta^s \sin \phi^s$, we have
  \begin{align*}
    N^s_{\theta} =
     & (1-\Gamma) \frac{ E^{i} }{\eta} \cos \theta^i \cos \theta^s \bigl( \cos \phi^i \sin \phi^s -\sin \phi^i \cos \phi^s \bigr) \\
     & \int_{-\frac{a}{2}}^{\frac{a}{2}} e^{ j 2 \pi x ( \sin \theta^s \cos \phi^s + \sin \theta^i \cos \phi^i ) / \lambda } d x 
    \int_{-\frac{b}{2}}^{\frac{b}{2}} e^{ j 2 \pi y ( \sin \theta^s \sin \phi^s + \sin \theta^i \sin \phi^i ) / \lambda } d y
  \end{align*}
  Using the integral $\int_{-\frac{c}{2}}^{\frac{c}{2}} e^{j \alpha x} d x = c \biggl[ \frac{\sin \bigl( \frac{\alpha}{2} c \bigr) }{ \frac{\alpha}{2} c } \biggr]$, the above expression could be written as
  \begin{align*}
    N^s_{\theta} =
     & (1-\Gamma)  a b \frac{ E^{i} }{\eta} \cos \theta^i \cos \theta^s \bigl( \cos \phi^i \sin \phi^s - \sin \phi^i \cos \phi^s \bigr) 
     \text{Sa} (a, b; \theta^s,\phi^s; \theta^i, \phi^i) .
  \end{align*}
  In the same vein, we have
  \begin{align*}
    N^s_{\phi} 
    = & \iint_{S} \bigl[ - J^s_x \sin \phi^s + J^s_y \cos \phi^s \bigr] e^{j 2 \pi r' \cos \psi / \lambda } d \Omega  \\
    = & (1-\Gamma) a b \frac{ E^{i} }{\eta} \cos \theta^i \bigl( \sin \phi^i \sin \phi^s + \cos \phi^i \cos \phi^s  \bigr) \text{Sa} (a, b; \theta^s,\phi^s; \theta^i, \phi^i) .
  \end{align*}

Then, in the far-field regime, the $(r,\theta,\phi)$ components of the scattered electric field $\mathbf{E}^s$ are \cite[eq.~(6-122)]{balanis2012advanced}
  \begin{equation*}
    \begin{dcases*}
      \begin{aligned}
      E^s_r ( r^s, \theta^s, \phi^s ) \simeq 
      & 0 , \\
      E^s_{\theta} ( r^s, \theta^s, \phi^s ) \simeq 
      & C \frac{ A }{ \lambda } \frac{ e^{-j 2 \pi r^s / \lambda } }{ r^s } E^i \cos \theta^i \cos \theta^s \bigl( \cos \phi^i \sin \phi^s - \sin \phi^i \cos\phi^s\bigr) \text{Sa} (a, b; \theta^s,\phi^s; \theta^i, \phi^i) , \\
      E^s_{\phi} ( r^s, \theta^s, \phi^s ) \simeq 
      & C \frac{ A }{ \lambda } \frac{ e^{-j 2 \pi r^s / \lambda } }{ r^s } E^i \cos \theta^i \bigl( \sin \phi^i \sin \phi^s + \cos \phi^i \cos \phi^s \bigr) \text{Sa} (a, b; \theta^s,\phi^s; \theta^i, \phi^i) .
      \end{aligned}
    \end{dcases*}
  \end{equation*}
\end{proof}

\bibliographystyle{IEEEtran}
\bibliography{References}

\begin{thebibliography}{10}
\providecommand{\url}[1]{#1}
\csname url@samestyle\endcsname
\providecommand{\newblock}{\relax}
\providecommand{\bibinfo}[2]{#2}
\providecommand{\BIBentrySTDinterwordspacing}{\spaceskip=0pt\relax}
\providecommand{\BIBentryALTinterwordstretchfactor}{4}
\providecommand{\BIBentryALTinterwordspacing}{\spaceskip=\fontdimen2\font plus
\BIBentryALTinterwordstretchfactor\fontdimen3\font minus
  \fontdimen4\font\relax}
\providecommand{\BIBforeignlanguage}[2]{{%
\expandafter\ifx\csname l@#1\endcsname\relax
\typeout{** WARNING: IEEEtran.bst: No hyphenation pattern has been}%
\typeout{** loaded for the language `#1'. Using the pattern for}%
\typeout{** the default language instead.}%
\else
\language=\csname l@#1\endcsname
\fi
#2}}
\providecommand{\BIBdecl}{\relax}
\BIBdecl

\bibitem{basar2019wireless}
E.~Basar, M.~Di~Renzo, J.~De~Rosny, M.~Debbah, M.-S. Alouini, and R.~Zhang,
  ``Wireless communications through reconfigurable intelligent surfaces,''
  \emph{IEEE Access}, vol.~7, pp. 116\,753--116\,773, 2019.

\bibitem{wu2019intelligent}
Q.~Wu and R.~Zhang, ``Intelligent reflecting surface enhanced wireless network
  via joint active and passive beamforming,'' \emph{IEEE Trans. Wireless
  Commun.}, vol.~18, no.~11, pp. 5394--5409, 2019.

\bibitem{wu2019towards}
------, ``Towards smart and reconfigurable environment: intelligent reflecting
  surface aided wireless network,'' \emph{IEEE Commun. Mag.}, vol.~58, no.~1,
  pp. 106--112, 2019.

\bibitem{yuan2021reconfigurable}
X.~Yuan, Y.-J.~A. Zhang, Y.~Shi, W.~Yan, and H.~Liu,
  ``Reconfigurable-intelligent-surface empowered wireless communications:
  challenges and opportunities,'' \emph{IEEE Wireless Commun.}, vol.~28, no.~2,
  pp. 136--143, 2021.

\bibitem{wymeersch2020radio}
H.~Wymeersch, J.~He, B.~Denis, A.~Clemente, and M.~Juntti, ``Radio localization
  and mapping with reconfigurable intelligent surfaces: challenges,
  opportunities, and research directions,'' \emph{IEEE Veh. Technol. Mag.},
  vol.~15, no.~4, pp. 52--61, 2020.

\bibitem{liu2018programmable}
F.~Liu, A.~Pitilakis, M.~S. Mirmoosa, O.~Tsilipakos, X.~Wang, A.~C.
  Tasolamprou, S.~Abadal, A.~Cabellos-Aparicio, E.~Alarc{\'o}n, C.~Liaskos
  \emph{et~al.}, ``Programmable metasurfaces: state of the art and prospects,''
  in \emph{Proc. IEEE Int. Symp. Circuits Syst. (ISCAS)}.\hskip 1em plus 0.5em
  minus 0.4em\relax IEEE, 2018, pp. 1--5.

\bibitem{pei2021ris}
X.~Pei, H.~Yin, L.~Tan, L.~Cao, Z.~Li, K.~Wang, K.~Zhang, and E.~Bj{\"o}rnson,
  ``{RIS}-aided wireless communications: prototyping, adaptive beamforming, and
  indoor/outdoor field trials,'' \emph{IEEE Trans. Commun.}, vol.~69, no.~12,
  pp. 8627--8640, 2021.

\bibitem{tang2020wireless}
W.~Tang, M.~Z. Chen, X.~Chen, J.~Y. Dai, Y.~Han, M.~Di~Renzo, Y.~Zeng, S.~Jin,
  Q.~Cheng, and T.~J. Cui, ``Wireless communications with reconfigurable
  intelligent surface: path loss modeling and experimental measurement,''
  \emph{IEEE Trans. Wireless Commun.}, vol.~20, no.~1, pp. 421--439, 2020.

\bibitem{dai2020reconfigurable}
L.~Dai, B.~Wang, M.~Wang, X.~Yang, J.~Tan, S.~Bi, S.~Xu, F.~Yang, Z.~Chen,
  M.~Di~Renzo \emph{et~al.}, ``Reconfigurable intelligent surface-based
  wireless communications: antenna design, prototyping, and experimental
  results,'' \emph{IEEE Access}, vol.~8, pp. 45\,913--45\,923, 2020.

\bibitem{nayeri2018reflectarray}
P.~Nayeri, F.~Yang, and A.~Z. Elsherbeni, \emph{Reflectarray antennas: theory,
  designs, and applications}.\hskip 1em plus 0.5em minus 0.4em\relax John Wiley
  \& Sons, 2018.

\bibitem{di2020smart}
M.~Di~Renzo, A.~Zappone, M.~Debbah, M.-S. Alouini, C.~Yuen, J.~De~Rosny, and
  S.~Tretyakov, ``Smart radio environments empowered by reconfigurable
  intelligent surfaces: how it works, state of research, and the road ahead,''
  \emph{IEEE J. Sel. Areas Commun.}, vol.~38, no.~11, pp. 2450--2525, 2020.

\bibitem{di2020reconfigurable}
M.~Di~Renzo, K.~Ntontin, J.~Song, F.~H. Danufane, X.~Qian, F.~Lazarakis,
  J.~De~Rosny, D.-T. Phan-Huy, O.~Simeone, R.~Zhang \emph{et~al.},
  ``Reconfigurable intelligent surfaces vs. relaying: Differences,
  similarities, and performance comparison,'' \emph{IEEE Open J. Commun. Soc.},
  vol.~1, pp. 798--807, 2020.

\bibitem{huang2020holographic}
C.~Huang, S.~Hu, G.~C. Alexandropoulos, A.~Zappone, C.~Yuen, R.~Zhang,
  M.~Di~Renzo, and M.~Debbah, ``Holographic {MIMO} surfaces for 6g wireless
  networks: opportunities, challenges, and trends,'' \emph{IEEE Wireless
  Commun.}, vol.~27, no.~5, pp. 118--125, 2020.

\bibitem{liu2020matrix}
H.~Liu, X.~Yuan, and Y.-J.~A. Zhang, ``Matrix-calibration-based cascaded
  channel estimation for reconfigurable intelligent surface assisted multiuser
  {MIMO},'' \emph{IEEE J. Sel. Areas Commun.}, vol.~38, no.~11, pp. 2621--2636,
  2020.

\bibitem{wei2021channel}
L.~Wei, C.~Huang, G.~C. Alexandropoulos, C.~Yuen, Z.~Zhang, and M.~Debbah,
  ``Channel estimation for {RIS}-empowered multi-user {MISO} wireless
  communications,'' \emph{IEEE Trans. Commun.}, vol.~69, no.~6, pp. 4144--4157,
  2021.

\bibitem{zheng2020intelligent}
B.~Zheng, C.~You, and R.~Zhang, ``Intelligent reflecting surface assisted
  multi-user {OFDMA}: channel estimation and training design,'' \emph{IEEE
  Trans. Wireless Commun.}, vol.~19, no.~12, pp. 8315--8329, 2020.

\bibitem{faqiri2022physfad}
R.~Faqiri, C.~Saigre-Tardif, G.~C. Alexandropoulos, N.~Shlezinger, M.~F. Imani,
  and P.~del Hougne, ``Physfad: physics-based end-to-end channel modeling of
  {RIS}-parametrized environments with adjustable fading,'' \emph{arXiv
  preprint arXiv:2202.02673}, 2022.

\bibitem{yildirim2021hybrid}
I.~Yildirim, F.~Kilinc, E.~Basar, and G.~C. Alexandropoulos, ``Hybrid
  {RIS}-empowered reflection and decode-and-forward relaying for coverage
  extension,'' \emph{IEEE Commun. Lett.}, vol.~25, no.~5, pp. 1692--1696, 2021.

\bibitem{zheng2021double}
B.~Zheng, C.~You, and R.~Zhang, ``Double-{IRS} assisted multi-user {MIMO}:
  cooperative passive beamforming design,'' \emph{IEEE Trans. Wireless
  Commun.}, vol.~20, no.~7, pp. 4513--4526, 2021.

\bibitem{di2020hybrid}
B.~Di, H.~Zhang, L.~Song, Y.~Li, Z.~Han, and H.~V. Poor, ``Hybrid beamforming
  for reconfigurable intelligent surface based multi-user communications:
  achievable rates with limited discrete phase shifts,'' \emph{IEEE J. Sel.
  Areas Commun.}, vol.~38, no.~8, pp. 1809--1822, 2020.

\bibitem{zhang2020reconfigurable}
H.~Zhang, B.~Di, L.~Song, and Z.~Han, ``Reconfigurable intelligent surfaces
  assisted communications with limited phase shifts: How many phase shifts are
  enough?'' \emph{IEEE Trans. Veh. Technol.}, vol.~69, no.~4, pp. 4498--4502,
  2020.

\bibitem{zhou2020spectral}
S.~Zhou, W.~Xu, K.~Wang, M.~Di~Renzo, and M.-S. Alouini, ``Spectral and energy
  efficiency of {IRS}-assisted {MISO} communication with hardware
  impairments,'' \emph{IEEE Wireless Commun. Lett.}, vol.~9, no.~9, pp.
  1366--1369, 2020.

\bibitem{hu2018beyond1}
S.~Hu, F.~Rusek, and O.~Edfors, ``Beyond massive {MIMO}: the potential of
  positioning with large intelligent surfaces,'' \emph{IEEE Trans. Signal
  Process.}, vol.~66, no.~7, pp. 1761--1774, 2018.

\bibitem{zhang2022holographic}
H.~Zhang, H.~Zhang, B.~Di, M.~Di~Renzo, Z.~Han, H.~V. Poor, and L.~Song,
  ``Holographic integrated sensing and communication,'' \emph{IEEE J. Sel.
  Areas Commun.}, 2022.

\bibitem{puglielli2015design}
A.~Puglielli, A.~Townley, G.~LaCaille, V.~Milovanovi{\'c}, P.~Lu,
  K.~Trotskovsky, A.~Whitcombe, N.~Narevsky, G.~Wright, T.~Courtade
  \emph{et~al.}, ``Design of energy-and cost-efficient massive {MIMO} arrays,''
  \emph{Proc. IEEE}, vol. 104, no.~3, pp. 586--606, 2015.

\bibitem{long2021active}
R.~Long, Y.-C. Liang, Y.~Pei, and E.~G. Larsson, ``Active reconfigurable
  intelligent surface-aided wireless communications,'' \emph{IEEE Trans.
  Wireless Commun.}, vol.~20, no.~8, pp. 4962--4975, 2021.

\bibitem{song2022intelligent}
X.~Song, J.~Xu, F.~Liu, T.~X. Han, and Y.~C. Eldar, ``Intelligent reflecting
  surface enabled sensing: {C}ram\'er-{R}ao lower bound optimization,''
  \emph{arXiv:2204.11071}, 2022.

\bibitem{ozdogan2019intelligent}
{\"O}.~{\"O}zdogan, E.~Bj{\"o}rnson, and E.~G. Larsson, ``Intelligent
  reflecting surfaces: physics, propagation, and pathloss modeling,''
  \emph{IEEE Wireless Commun. Lett.}, vol.~9, no.~5, pp. 581--585, 2019.

\bibitem{najafi2020physics}
M.~Najafi, V.~Jamali, R.~Schober, and H.~V. Poor, ``Physics-based modeling and
  scalable optimization of large intelligent reflecting surfaces,'' \emph{IEEE
  Trans. Commun.}, vol.~69, no.~4, pp. 2673--2691, 2020.

\bibitem{bjornson2020power}
E.~Bj{\"o}rnson and L.~Sanguinetti, ``Power scaling laws and near-field
  behaviors of massive {MIMO} and intelligent reflecting surfaces,'' \emph{IEEE
  Open J. Commun. Soc.}, vol.~1, pp. 1306--1324, 2020.

\bibitem{di2021communication}
M.~Di~Renzo, F.~H. Danufane, and S.~Tretyakov, ``Communication models for
  reconfigurable intelligent surfaces: From surface electromagnetics to
  wireless networks optimization,'' \emph{arXiv:2110.00833}, 2021.

\bibitem{ellingson2021path}
S.~W. Ellingson, ``Path loss in reconfigurable intelligent surface-enabled
  channels,'' in \emph{Proc. IEEE PIMRC}.\hskip 1em plus 0.5em minus
  0.4em\relax IEEE, 2021, pp. 829--835.

\bibitem{kammoun2020asymptotic}
A.~Kammoun, A.~Chaaban, M.~Debbah, M.-S. Alouini \emph{et~al.}, ``Asymptotic
  max-min {SINR} analysis of reconfigurable intelligent surface assisted {MISO}
  systems,'' \emph{IEEE Trans. Wireless Commun.}, vol.~19, no.~12, pp.
  7748--7764, 2020.

\bibitem{gradoni2021end}
G.~Gradoni and M.~Di~Renzo, ``End-to-end mutual coupling aware communication
  model for reconfigurable intelligent surfaces: an electromagnetic-compliant
  approach based on mutual impedances,'' \emph{IEEE Wireless Commun. Lett.},
  vol.~10, no.~5, pp. 938--942, 2021.

\bibitem{abeywickrama2020intelligent}
S.~Abeywickrama, R.~Zhang, Q.~Wu, and C.~Yuen, ``Intelligent reflecting
  surface: Practical phase shift model and beamforming optimization,''
  \emph{IEEE Trans. Commun.}, vol.~68, no.~9, pp. 5849--5863, 2020.

\bibitem{balanis2012advanced}
C.~A. Balanis, \emph{Advanced engineering electromagnetics}.\hskip 1em plus
  0.5em minus 0.4em\relax John Wiley \& Sons, 2012.

\bibitem{tapio2021survey}
V.~Tapio, I.~Hemadeh, A.~Mourad, A.~Shojaeifard, and M.~Juntti, ``Survey on
  reconfigurable intelligent surfaces below 10 {GHz},'' \emph{EURASIP J.
  Wireless Commun. Netw.}, vol. 2021, no.~1, pp. 1--18, 2021.

\bibitem{dai2019wireless}
J.~Y. Dai, W.~K. Tang, J.~Zhao, X.~Li, Q.~Cheng, J.~C. Ke, M.~Z. Chen, S.~Jin,
  and T.~J. Cui, ``Wireless communications through a simplified architecture
  based on time-domain digital coding metasurface,'' \emph{Adv. Mater.
  Technol.}, vol.~4, no.~7, p. 1900044, 2019.

\bibitem{wang2019design}
M.~Wang, S.~Xu, F.~Yang, and M.~Li, ``Design and measurement of a 1-bit
  reconfigurable transmitarray with subwavelength {H}-shaped coupling slot
  elements,'' \emph{IEEE Trans. Antennas Propag.}, vol.~67, no.~5, pp.
  3500--3504, 2019.

\bibitem{osipov2017modern}
A.~V. Osipov and S.~A. Tretyakov, \emph{Modern electromagnetic scattering
  theory with applications}.\hskip 1em plus 0.5em minus 0.4em\relax John Wiley
  \& Sons, 2017.

\bibitem{pan2016bad}
V.~Y. Pan, ``How bad are {V}andermonde matrices?'' \emph{SIAM J. Matrix Anal.
  Appl.}, vol.~37, no.~2, pp. 676--694, 2016.

\bibitem{liu2021reconfigurable}
Y.~Liu, X.~Liu, X.~Mu, T.~Hou, J.~Xu, M.~Di~Renzo, and N.~Al-Dhahir,
  ``Reconfigurable intelligent surfaces: principles and opportunities,''
  \emph{IEEE Commun. Surv. Tut.}, vol.~23, no.~3, pp. 1546--1577, 2021.

\bibitem{ding2020impact}
Z.~Ding, R.~Schober, and H.~V. Poor, ``On the impact of phase shifting designs
  on {IRS}-{NOMA},'' \emph{IEEE Wireless Commun. Lett.}, vol.~9, no.~10, pp.
  1596--1600, 2020.

\bibitem{abu2021near}
Z.~Abu-Shaban, K.~Keykhosravi, M.~F. Keskin, G.~C. Alexandropoulos,
  G.~Seco-Granados, and H.~Wymeersch, ``Near-field localization with a
  reconfigurable intelligent surface acting as lens,'' in \emph{Proc. IEEE Int.
  Conf. Commun. (ICC)}.\hskip 1em plus 0.5em minus 0.4em\relax IEEE, 2021, pp.
  1--6.

\bibitem{zhao2021exploiting}
M.-M. Zhao, Q.~Wu, M.-J. Zhao, and R.~Zhang, ``Exploiting amplitude control in
  intelligent reflecting surface aided wireless communication with imperfect
  {CSI},'' \emph{IEEE Trans. Commun.}, vol.~69, no.~6, pp. 4216--4231, 2021.

\bibitem{liu2014broadband}
L.~Liu, X.~Zhang, M.~Kenney, X.~Su, N.~Xu, C.~Ouyang, Y.~Shi, J.~Han, W.~Zhang,
  and S.~Zhang, ``Broadband metasurfaces with simultaneous control of phase and
  amplitude,'' \emph{Adv. Mater.}, vol.~26, no.~29, pp. 5031--5036, 2014.

\end{thebibliography}

\end{document}